\documentclass[11pt,scrartcl, a4paper]{article}
\usepackage[margin=1in]{geometry} 
\usepackage{setspace}
\usepackage{parskip}

\usepackage[affil-it]{authblk}
\usepackage{float}
\RequirePackage[colorlinks,linkcolor=black,citecolor=blue,urlcolor=blue]{hyperref}
\usepackage{amsfonts,amsthm,amsmath,mathrsfs,epsfig,natbib} 
\usepackage{multirow}
\usepackage{tabularx}
\usepackage{enumitem}
\usepackage{amssymb}
\usepackage{subcaption}
\usepackage{capt-of}
\usepackage{color,tikz}
\usepackage{cancel}
\usepackage{makecell}
\usepackage{xcolor}

\numberwithin{equation}{section}
\newcommand{\Var}{{\rm Var}}
\newcommand{\Cov}{{\rm Cov}}

\newcommand{\R}{\mathbb{R}}
\newcommand{\N}{\mathbb{N}}

\newcommand{\Q}{\mathbb{Q}}
\newcommand{\E}{\mathrm {E}}

\newtheorem{thm}{Theorem}

\newtheorem{rmk}[thm]{Remark}
\newtheorem{exa}[thm]{Example}
\newtheorem{assumption}{Assumption}
\newtheorem{defi}[thm]{Definition}

\def\AA{ \mathfrak{A} }

\def\SS{ \mathfrak{S} }

\newcommand\blfootnote[1]{%
  \begingroup
  \renewcommand\thefootnote{}\footnote{#1}%
  \addtocounter{footnote}{-1}%
  \endgroup
}

\makeatletter
\newcommand{\vast}{\bBigg@{4}}
\newcommand{\Vast}{\bBigg@{5}}
\makeatother

\begin{document}
 
\title{Jointly Exchangeable Collective Risk Models: Interaction, Structure, and Limit Theorems}

\author{Daniel Gaigall$^{1,2}$, Stefan Weber$^2$}

\affil{$^1$Institute for Data-driven Technologies, FH Aachen - University of Applied Sciences\\
$^2$House of Insurance, Leibniz University Hannover}

\date{\today}

\maketitle

\blfootnote{E-Mail: gaigall@fh-aachen.de (Daniel Gaigall); stefan.weber@insurance.uni-hannover.de (Stefan Weber)}

\begin{abstract}

We introduce a framework for systemic risk modeling in insurance portfolios using jointly exchangeable arrays, extending classical collective risk models to account for interactions. {The Aldous-Hoover-Kallenberg representation characterizes joint exchangeability, a more general probabilistic symmetry than de Finetti's exchangeability.} We establish central limit theorems that asymptotically capture total portfolio losses, providing a theoretical foundation for approximations in large portfolios and over long time horizons. These approximations are validated  through simulation-based numerical experiments. Additionally, we analyze the impact of dependence on portfolio loss distributions, with a particular focus on tail behavior.
\end{abstract}

\noindent \textit{Keywords:}   Central limit theorem; Collective risk model; Joint Exchangeability; Systemic risk; Contagion; Dependence; Insurance mathematics; Loss distributions

\section{Introduction}\label{int}

Systemic risk is a major concern in finance and insurance. Idiosyncratic and systematic risks account for only part of the probabilistic variations in outcomes, while systemic interactions play a crucial role in financial markets, cyber networks, power grids, and other infrastructures. This paper focuses on modeling insurance portfolio losses within a generalized collective model, or frequency-severity model. Classical risk theory does not inherently capture systemic effects, highlighting the well-recognized need for more sophisticated models.

We introduce a probabilistic model for total losses, such as those of an insurance company, using the theory of jointly exchangeable arrays, a weaker notion than classical exchangeability. This approach extends beyond idiosyncratic and systematic risk by incorporating a specific systemic component: network interaction. Specifically, our framework generalizes de Finetti’s exchangeability by preserving network interactions while allowing for relabeling, unlike conventional models that assume homogeneity after adjusting for covariates. Example specifications include random graphs and graphon models -- such as Erd\H{o}s-R{\'e}nyi random networks -- to describe transmissions and infections within an insurance portfolio. Additionally, conditional loss expenditures are modeled as dependent through jointly exchangeable arrays. {Our generalized model also accounts for external sources, systemic background shocks and  time-dependent interactions.} The classical collective risk model emerges as a special case. Within this framework, we establish several limit theorems that asymptotically characterize the distribution of total portfolio losses{, and we derive bounds on the rates of convergence.}

In insurance, systemic risk naturally arises in networked systems. Relevant examples include power grids and digital infrastructures, operating both at local and global scales, as well as contagion-type dynamics generated by cyber threats. Market sentiment can likewise be modeled within this framework. Recent surveys addressing these phenomena include \citet{aw-s23} and \citet{ENISA24}. Our approach is most closely related to the systemic perspective of \citet{zs22}, while placing particular emphasis on underlying probabilistic symmetries and on the derivation of corresponding limit theorems.
 
For a systematic presentation of probabilistic symmetries, including joint exchangeability, we refer to \citet{ald} and \citet{kal3}. Under the additional assumption that jointly exchangeable arrays are dissociated, limit results in distribution for related statistics are derived in \citet{sil}, with further results in \citet{eag}. Jointly exchangeable distributions admit an ergodic decomposition that extends de Finetti's theorem{. See} \citet{dF37}. These structural results are developed in \citet{hoo}, \citet{ald0}, \citet{kal1}, \citet{kal2}, and \citet{kal3}. The Choquet representation of the joint distribution as a mixture of dissociated arrays provides a natural reduction to this case.

Limit theorems for exchangeable arrays have been studied from various perspectives. \citet{eag2} {examine} dissociated arrays and generalizations of central limit theorems to $U$-statistics. \citet{dav} {extend} these results to empirical processes. \citet{aus} {develop} a general theory of limit theorems within the framework of distributional symmetries defined by invariance under group transformations. Berry-Esseen type results for generalized $U$-statistics are provided in \citet{zha}.  Although this broader theory applies in our setting, we use the classical results of \citet{sil} and \citet{eag}, which are sufficient for our analysis. 

Building on these theoretical foundations, we now turn to the implications for risk management and loss distribution modeling. Of particular interest are the properties of the underlying loss distribution. In our generalized collective model with interacting losses, the expectation and variance of the total loss of the insurance company can be explicitly computed. We establish asymptotic characterizations of the total loss distribution as the number of contracts or the time horizon grows. The central limit theorems we establish provide a useful tool for risk management and include Anscombe-type results; see \citet{ans}. Simulations demonstrate that these asymptotic approximations remain accurate even for moderate portfolio sizes and relatively short time horizons.

The use of probabilistic interacting particle models to capture systemic risk in economic applications dates back to \citet{foe74}. In credit risk, \citet{giwe04, giwe06} asymptotically characterize total portfolio losses under local interaction, specifically employing a voter model. In insurance, the collective model with systemic effects is particularly relevant for cyber risk and associated losses. For surveys on cyber insurance, see \citet{aw-s23}, \citet{dakr23}, \citet{el20}, \citet{ENISA24}, and \citet{kswz2024}. Closely related to our approach, \citet{zs22} study a cyber collective model incorporating covariates but without network interactions. Other forms of local and global interactions in actuarial cyber risk models are studied in \citet{fww18}, \citet{hl21}, \citet{hetal22}, \citet{bbh21}, \citet{hetal23}, and \citet{aw24}, though without addressing limit theorems under joint exchangeability.

\noindent The main contributions of this paper are
\begin{itemize}
\item[(i)] a novel framework for systemic risk modeling in insurance portfolios based on jointly exchangeable arrays, extending classical collective risk models to incorporate interactions,
\item[(ii)] asymptotic characterizations of total losses via central limit theorems  and  related rates of convergence,  providing theoretical justification for approximations in large portfolios and over long time horizons,
\item[(iii)] simulation-based  validation of the model and its approximations, demonstrating the practical relevance of the approach for risk management applications.
\end{itemize}

The structure of the paper is as follows. Section~\ref{modint} introduces the extended collective model, discusses joint exchangeability, and provides illustrative examples. Section~\ref{theory} presents limit theorems for total losses. {Section~\ref{rel} characterizes the relation to classical collective models.} All proofs are deferred to Appendix~\ref{proofs}.  In Section~\ref{impsim}, we conduct case studies, validate the accuracy of the approximations, and analyze the impact of dependence. Finally, Section~\ref{con} concludes with a discussion of open research questions.

\section{A collective loss model with interaction}\label{modint}

In this paper, we are interested in a frequency-severity model of insurance losses, {also called a collective model}, which includes a special form of contagious interactions. We will first explain the type of interaction on which we focus {before we describe} the full dynamic model. The underlying probability space will be denoted by $(\Omega,\AA,P)$ and assumed to be atomless. Moreover, we allow extensions of the probability space  to ensure  almost sure representations introduced later.

\subsection{Contagious losses}\label{con0}

We consider a collection of agents or, in a collective model, of loss occurrences, enumerated by the natural numbers $\N$. In the latter case, unlike some common conventions, losses of zero are also admissible. Throughout the paper, the terms agent and entity will be used interchangeably. We assume that agents can trigger losses of others, a phenomenon that can be captured by a graph as illustrated in Figure \ref{fig0}. Vertices represent entities, while edges correspond to idiosyncratic infections and transmission pathways. Infections and transmissions occur according to a random mechanism. {The transmissions can be interpreted as endogenous effects, while the infections can be seen as exogenous effects.}

\begin{figure}[h]
     \centering
         \includegraphics[width=0.4\linewidth]{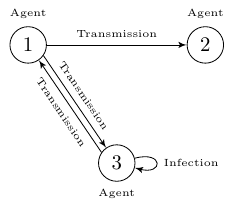}
       \caption{Agents, transmissions and infections as a directed graph.}
        \label{fig0}
\end{figure}

The random graph can be encoded by an infinite matrix (also called array) $I=(I_{i,j})_{(i,j) \in\N\times \N}$ with entries that are $\{0,1\}$-valued random variables. For $i\neq j$,  entity $j$ causes a loss at entity $i$, if $I_{i,j} = 1$; otherwise, $I_{i,j} = 0$. This definition does not require entity $j$ itself  to be infected. An infection of agent $i$ is represented by $I_{i,i} = 1$; otherwise, $I_{i,i} = 0$.
Loss expenditures resulting from infections and transmission are specified by an array $Z=(Z_{i,j})_{(i,j) \in\N\times \N}$ whose entries are  $[0,\infty)$-valued random variables. For $i,j\in\N$,  the random variable $Z_{i,j}$  captures the loss expenditure of policy holder $i$ caused either by a transmission from entity $j$ to agent $i$ (if $i\neq j$) or by an infection of the entity itself (if $i= j$).

The resulting losses are given by the array $G=(G_{i,j})_{(i,j) \in\N\times \N}$, where
\begin{equation}\label{g:def}
G_{i,j}=I_{i,j} \cdot Z_{i,j} , \quad i,j \in \N.
\end{equation}
In other words, $G_{i,j}$ represents the actual loss incurred by  entity $i$  due to a transmission from agent $j$ to $i$ (if $i\neq j$) or due to an infection of the entity itself (if $i= j$).

\subsection{Jointly exchangeable arrays}\label{exc}

In this paper, we impose structure on the random mechanisms that generate the losses. Instead of relying on the classical notion of independence, we consider weaker concepts that allow for dependence. A fundamental, classical example is exchangeability, as introduced and characterized by Bruno de Finetti. While independence corresponds to the joint distribution being a product measure, exchangeability requires the joint distribution to be invariant under all finite permutations. De Finetti’s theorem shows that exchangeability is equivalent to conditional independence {with equal conditional distributions}, or, equivalently, to the joint distribution being a mixture of product measures {with identical marginals within each product}. This framework aligns precisely with the setting of classical Bayesian statistics.

We introduce an even weaker probabilistic symmetry, tailored to our context. By using the collection of edges as an index set, we consider joint exchangeability, which requires invariance of the joint distribution only under {the subgroup of edge permutations induced by vertex relabelings}. This relaxation accommodates more general dependence structures than exchangeability while preserving a tractable framework for analysis. For a comprehensive exposition on this and other distributional symmetries, see \citet{kal3}.

Although we focus on arrays of real-valued random variables, many definitions and results can often be stated in a more general form for random elements in measurable spaces. To ensure the existence of regular conditional distributions and simplify the application of standard results without additional verification, we assume that these spaces are Polish and maintain this assumption throughout.

\begin{defi}
Consider random elements $X_{i,j}$, $i,j \in\N$ in the Borel space $(R,\SS)$, defined on $(\Omega,\AA,P)$. The array $X=(X_{i,j})_{(i,j) \in\N\times \N}$ is called jointly exchangeable if it has the same distribution as $(X_{\pi(i),\pi(j)})_{(i,j) \in\N\times \N}$ for all bijective $\pi:\N\to\N$. 
\end{defi}
The definition can equivalently be expressed in terms of finite permutations. As in the exchangeable case, joint exchangeability also allows for a characterization of the structure of joint distributions in terms of mixtures of ergodic measures. An expository discussion can be found in \citet{orb}, while a detailed analysis is provided in \citet{kal3}. {The} ergodic measures can be characterized through sampling schemes. This characterization is most naturally expressed through functional representations.

To clarify the distinction between exchangeability and joint exchangeability, we begin by considering an exchangeable array $X = (X_{i,j})_{(i,j) \in \N \times \N}$, meaning that its joint distribution remains invariant under permutations of the edges $(i,j) \in \N \times \N$. A functional representation of De Finetti's theorem states that there exist i.i.d. random variables $\xi, \xi_{i,j}$, $i,j \in \N$, uniformly distributed in $[0,1]$, along with a measurable function $f:[0,1]^2 \to R$, such that almost surely 
$$X_{i,j} = f(\xi, \xi_{i,j}), \quad (i,j) \in \N \times \N.$$ 
Conditionally on $\xi$, the random elements $X_{i,j}$, $i,j \in \N$ are independent, with conditional distribution given by the probability kernel $\nu(\xi,B) = \int_{[0,1]}  I( f(\xi,z)\in B) dz,$ $B\in \SS$. In terms of this kernel and the distribution of $\xi$, the joint distribution can be expressed as a mixture of product measures, specifically, $\int_{[0,1]} \nu(z, \cdot)^{\otimes \infty} dz $.

The jointly exchangeable case (along with related symmetries such as contractibility or separate exchangeability) imposes a weaker invariance condition on the joint distribution, leading to a more general functional representation. As in the exchangeable case, we introduce i.i.d. random variables  $\xi_{i,j}$, $i,j \in \N$, though their role will now differ slightly. Consider a jointly exchangeable array $X = (X_{i,j})_{(i,j) \in \N \times \N}$, meaning that its joint distribution remains invariant under permutations of the vertices $i \in \N$. The array $X$ is not required to be symmetric. According to the Aldous-Hoover-Kallenberg theorem (see \citet{hoo}, \citet{ald0}, \citet{kal1}, \citet{kal2}) such an array admits the following functional representation: 

There exist i.i.d. random variables $\xi,\xi_i,\xi_{i,j}$, $i\leq j$, all uniformly distributed in $[0,1]$, along with a measurable function $h:[0,1]^4 \to R$, such that almost surely
$$X_{i,j} = h(\xi, \xi_{i}, \xi_{j}, \xi_{i,j}), \quad (i,j) \in \N \times \N,$$
where we set $\xi_{j,i} : = \xi_{i,j}$ for $i\leq j$. The latter symmetry assumption is an essential part of the functional representation, which characterizes the structure of the corresponding joint distribution through a sampling scheme. Conversely, any array that admits an Aldous-Hoover-Kallenberg representation is jointly exchangeable.

The distribution of the array can be decomposed as a mixture. To this end, we consider the conditional distribution of $X$ given $\xi$, represented by a kernel $\overline \nu$ from $[0,1]$ to $R^{\N\times\N}$. The original joint distribution is then characterized by the Choquet representation  $\int_{[0,1]} \overline \nu (z) dz$. For each $z\in[0,1]$, $\overline \nu (z)$ is the distribution of the array  $(h(z,  \xi_{i}, \xi_{j}, \xi_{i,j}))_{(i,j) \in\N\times \N}$, corresponding to the ergodic extremal points. These arrays are precisely those that are dissociated, a concept that we recall in the next definition.
\begin{defi}
Consider an array of random elements $X=(X_{i,j})_{(i,j) \in\N\times \N}$ taking values in $(R,\SS)$. We say that $X$ is dissociated if the subarrays $(X_{i,j})_{(i,j) \in A\times A}$ and $(X_{i,j})_{(i,j) \in B\times B}$ are independent whenever $A,B\subseteq\N$ and $A\cap B=\emptyset$.
\end{defi}

\begin{rmk}\label{graphon}

A special case of jointly exchangeable arrays is given by infinite graphs (see \citet{orb}). In this case, we set $R = \{0,1\}$. Focusing only on the case of undirected graphs, we obtain a functional representation by 
\begin{align*}
 X_{i,j}=I(\xi_{i,j}\le W(\xi,\xi_i,\xi_j)),~i,j\in\N,
\end{align*} 
where $W:[0,1]^3\rightarrow[0,1]$  is a measurable function that is symmetric in its last two arguments. In the dissociated case, there is no dependence on $\xi$, and the corresponding function is of the form $W:[0,1]^2\rightarrow[0,1]$. This function is often referred to as a graphon.
\end{rmk}

\subsection{Mixtures and limit theorems}\label{rmkdis}

For sums of jointly exchangeable arrays,
\begin{align*}
U(n)=\sum_{1\le i< j\le n} X_{i,j},
\end{align*}
where $X = (X_{i,j})_{(i,j) \in \N \times \N}$ takes values in the real numbers and $n \in \N$, {distributional limit results as $n\to\infty$ are available in} \citet{sil} and \citet{eag}. Limit theorems for the jointly exchangeable case follow directly from those for the dissociated case. Consider the functional representation
$X_{i,j} = h(\xi, \xi_{i}, \xi_{j}, \xi_{i,j})$, and define 
$
\widetilde X_{i,j}(z)=h(z,\xi_i,\xi_j,\xi_{i,j}),
$
 for $i,j\in\N$, with joint distribution $\bar \nu (z)$. Then, setting  $\widetilde U(n,z)=\sum_{1\le i< j\le n} \widetilde X_{i,j}(z)$, we obtain
\begin{align*}
P(U(n) \le x)=\int_{[0,1]}  P(\widetilde U(n,z)\le x)d z.
\end{align*}
This shows that distributional properties of sums related to jointly exchangeable arrays can be derived from those of jointly exchangeable and dissociated arrays via integration or mixture distributions.  {To pass from conditional to unconditional results, use dominated convergence for distributional limits and the triangle inequality for error bounds.} 

{We introduce the following assumption that will be used repeatedly.}

\begin{assumption}\label{a1}
The array of random variables $G$ is jointly exchangeable and satisfies $\E(G_{i,j}^4) < \infty$ for all $i,j \in \N$.
\end{assumption}

\subsection{Examples}\label{exa}

Due to the Aldous-Hoover-Kallenberg representation, constructing examples is straightforward. According to equation \eqref{g:def}, the array $G$  is determined by infections and transmissions $I$ and by conditional expenditures $Z$. If $I$  and $Z$ are independent of each other and satisfy either joint exchangeability, dissociatedness, or both, then $G$ inherits the respective property. We provide specifications for $I$ in Example~\ref{exa1} and for $Z$ in Example~\ref{exa2}; these arrays are jointly {exchangeable} and dissociated. The functional representation naturally extends to the non-dissociated case by introducing a common factor variable. As discussed, the general joint distribution can be decomposed into ergodic components -- i.e., dissociated distributions -- via a Choquet representation.

\begin{exa}\label{exa1}

a) Standard case: Infections (in the absence of transmissions) are  
\[
I \; = \; 
\begin{pmatrix}
J_{1} & 0 & 0 & \dots\\
0 & J_{2} & 0 & \dots\\
0 & 0 & J_{3} &  \dots \\
\vdots & \vdots & \vdots & \ddots
\end{pmatrix},
\]
where \( J_{i} \), \( i \in \N \), are independent and identically distributed \(\{0,1\}\)-valued random variables with \( P(J_{i} = 1) = p_J \) for some \( p_J \in [0,1] \). This example includes only infections but no transmissions. It corresponds to a classical insurance loss model. An illustration is provided in Figure~\ref{exfiga}.

b) Erd\H{o}s-R{\'e}nyi model: Transmissions (in the absence of infections) are modeled by 
\begin{align*}
I \; = \; \begin{pmatrix}
0 & K_{1}  & K_{2}  & \dots\\
K_{1} & 0  & K_{3}  & \dots\\
K_{2} & K_{3}  &0  & \dots\\
\vdots & \vdots & \vdots & \ddots
\end{pmatrix}.
\end{align*}
Here, the  $\{0,1\}$-valued random variables $K_{i}$, $i\in\N$, are independent and identically distributed with $P(K_{i}=1)=p_K$ for some $p_K\in[0,1]$. This example includes only symmetric transmissions, corresponding to an undirected random graph: the Erd\H{o}s-R{\'e}nyi model.  An illustration is provided in Figure~\ref{exfigb}.

c) Countermonotonic Erd\H{o}s-R{\'e}nyi model: This directed random graph is constructed from the undirected Erd\H{o}s-R{\'e}nyi model by tossing independent fair coins to assign a direction to each edge. Transmissions are given by 
\begin{align*}
I\; = \; \begin{pmatrix}
0 &J_{1} K_{1}  & J_{2}K_{2}  & \dots\\
(1-J_{1})K_{1} & 0  & J_{3}K_{3}  & \dots\\
(1-J_{2})K_{2} & (1-J_{3})K_{3}  &0  & \dots\\
\vdots & \vdots & \vdots & \ddots
\end{pmatrix}.
\end{align*}
Here, $J_{i}$, $K_{i}$, $i\in\N$, are independent  $\{0,1\}$-valued random variables. As before, $K_{i}$, $i\in\N$,   are identically distributed with  $P(K_{i}=1)=p_K$ for some $p_K\in[0,1]$. The ``fair coins''  $J_{i}$, $i\in\N$, are identically distributed with  $P(J_{i}=1)=1/2$. An illustration is presented in Figure~\ref{exfigc}. 

d) Erd\H{o}s-R{\'e}nyi model with infections: 
This case combines the standard case of infections with undirected Erd\H{o}s-R{\'e}nyi  transmissions, resulting in the array
\begin{align*}
I\;= \; \begin{pmatrix}
J_{1} & K_{1}  & K_{2}  & \dots\\
K_{1} & J_{2}  & K_{3}  & \dots\\
K_{2} & K_{3}  &J_{3}  & \dots\\
\vdots & \vdots & \vdots & \ddots
\end{pmatrix}.
\end{align*}
The independent $J_{i}$ and $K_{i}$, $i\in\N$, are $\{0,1\}$-valued with $P(J_{i}=1)=p_J$ for some $p_J\in[0,1]$ and  $P(K_{i}=1)=p_K$ for some $p_K\in[0,1]$. An illustration is provided in Figure~\ref{exfigd}.

e) Contagion model: In this case, entities can be infected idiosyncratically. {Only infected entities can transmit. Conditional on the infection indicators, transmissions are independent.} This is captured by the array 
\begin{align*}
I \;= \; \begin{pmatrix}
J_{1} & J_{2}K_{1,2}  & J_{3}K_{1,3}  & \dots\\
J_{1}K_{2,1} & J_{2}  & J_{3}K_{2,3}  & \dots\\
J_{1}K_{3,1} & J_{2}K_{3,2}  & J_{3}  & \dots\\
\vdots & \vdots & \vdots & \ddots
\end{pmatrix}.
\end{align*}
The independent $J_{i}$ and $K_{i,j}$, $i,j\in\N$,  $i\neq j$, are $\{0,1\}$-valued with $P(J_{i}=1)=p_J$ for some $p_J\in[0,1]$ and  $P(K_{i,j}=1)=p_K$ for some $p_K\in[0,1]$. An illustration is provided in Figure~\ref{exfige}.

\begin{figure}[h]
     \centering
     \begin{subfigure}[b]{0.35\textwidth}
         \centering
         \includegraphics[width=1\linewidth]{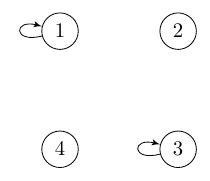}
         \caption{Standard case}
         \label{exfiga}
     \end{subfigure}%
\hspace{2cm}
     \begin{subfigure}[b]{0.3\textwidth}
         \centering
         \includegraphics[width=1\linewidth]{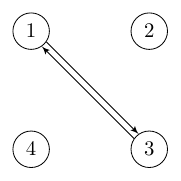}
         \caption{Erd\H{o}s-R{\'e}nyi model}
         \label{exfigb}
     \end{subfigure}

\hspace{0.7cm}
    \begin{subfigure}[b]{0.3\textwidth}
         \centering
    \includegraphics[width=1\linewidth]{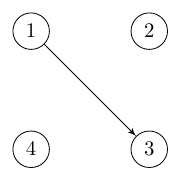}
         \caption{Countermonotonicity}
         \label{exfigc}
     \end{subfigure}%
\hspace{2cm}
     \begin{subfigure}[b]{0.3\textwidth}
         \centering
         \includegraphics[width=1\linewidth]{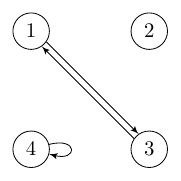}
         \caption{Erd\H{o}s-R{\'e}nyi with infections}
         \label{exfigd}
     \end{subfigure}
    \begin{subfigure}[b]{0.35\textwidth}
         \centering
    \includegraphics[width=1\linewidth]{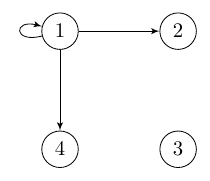}
         \caption{Contagion model}
         \label{exfige}
     \end{subfigure}%
       \caption{Transmissions and infections in the model specifications in Example \ref{exa1}.}
        \label{fig1}
\end{figure}
\end{exa}

Next, we consider specifications of conditional loss expenditures, focusing again on jointly exchangeable and dissociated examples.

\begin{exa}\label{exa2}
a) Independent conditional losses: If the loss expenditures $Z_{i,j}$, $i,j\in\N$, are independent and identically distributed, the model closely resembles the classical collective risk model. However, systemic risk typically requires incorporating dependencies, generalizing independence. This can be achieved, for example, by imposing only joint exchangeability, based on the Aldous-Hoover-Kallenberg representation.

b) {Comonotonic conditional loss expenditures}: The random variable $Z_{i,j}$ represents the conditional loss expenditure from entity $j$ to $i$. A specific type of contagion can be modeled by assuming that these losses depend only on their origin, {so for each fixed source $j$, the losses $(Z_{i,j})_{i\in\N}$ are equal and hence comonotonic. Specifically, take} positive random variables $\widetilde Z_j$, $j\in\N$, assumed to be i.i.d., and an increasing, measurable function $h:[0,\infty) \to [0,\infty)$ such that $Z_{i,j} = h(\widetilde Z_j)$. In this simplest case $h(x) = x$ for all $x$, so that the conditional loss expenditures satisfy $Z_{i,j} = \widetilde Z_j$ for all $i,j \in \N$.

c) Positively dependent conditional {loss expenditures}: More generally, we consider a measurable function $h: [0,\infty) \times [0,1] \to [0,\infty)$ that is increasing in the first argument. Conditional losses are given by
\begin{align*}
Z_{i,j}=h(\widetilde Z_{j},\varepsilon_{i,j}),~i,j\in\N,
\end{align*}
where the i.i.d. and  positive random variables $\widetilde Z_{j}$, $j\in\N$, model the contagious component of the loss expenditures, while the independent shocks $\varepsilon_{i,j}$, $i,j\in\N$, uniformly distributed in $[0,1]$, capture the idiosyncratic properties of {loss expenditures}. Of course, in specific models, one could alternatively choose different functions $h$ and distributions for $\varepsilon_{i,j}$, where appropriate.
\end{exa}

\subsection{The number of loss events}\label{number}

 \color{black}
To model  a dynamic insurance portfolio, we introduce a counting process $N=(N(t))_{t\in [0,\infty)}$, {i.e.,} a stochastic process indexed by $[0,\infty)$ satisfying $N(0)=0$,   with values in $\N_0$ and having right-continuous and non-decreasing sample paths. For a given time $t\in [0,\infty)$, $N(t)$ represents the number of loss events up to time $t$.  We have the following representation of the counting process. Setting $\inf\emptyset=+\infty$, the random variables
$$T_k=\inf\{t\in[0,\infty); N(t)\ge k\},~k\in\N,$$ 
are the time points of  loss events that satisfy $T_1\le T_2 \le \dots$ and take values in $[0,\infty]$. Conversely, if $T_k$, $k\in\N$, are the time points of  loss events, we obtain our counting process via 
$$N(t)=|\{k\in\N;T_k\le t\}|,~t\in[0,\infty).$$
In particular, the counting process can be represented by the {sequence} of time points of loss events.  Denote by $F_{T_1}$ the distribution function of $T_1$ and by $F_{T_k|T_1=t_1,\dots,T_{k-1}=t_{k-1}}$ the conditional distribution function of $T_k$ given $T_1=t_1,\dots,T_{k-1}=t_{k-1}$, $t_1,\dots,t_{k-1}\in[0,\infty]$, $k\ge 2$.  Define measurable maps $\ell_1:[0,1]\rightarrow  [0,\infty]$ and $\ell_k:[0,1]\times [0,\infty]^{k-1}\rightarrow  [0,\infty]$, $k\ge 2$,  by $\ell_1(u)= F_{T_1}^{-1}(u)$, $u\in [0,1]$, and $\ell_k(u,t_1,\dots,t_{k-1})= F_{T_k|T_1=t_1,\dots,T_{k-1}=t_{k-1}}^{-1}(u)$, $u\in [0,1]$, $t_1,\dots,t_{k-1}\in[0,\infty]$, $k\ge 2$. It follows from the (inverse) Rosenblatt transformation that there are   i.i.d. random variables $\zeta_k$, $k\in\N$, all uniformly distributed in $[0,1]$,  such that almost surely
$$T_1=\ell_1(\zeta_1),~T_k=\ell_k(\zeta_k,T_1,\dots,T_{k-1}),~k\ge 2.$$
For technical reasons, we state the following moment assumption.

\begin{assumption}\label{a2}
The counting process $N$ satisfies $\E(N(t)^2)<\infty$ for all $t\in[0,\infty)$.
\end{assumption}
Various concrete specifications of $N$ are possible.

\begin{exa}\label{exatime}
a) A common example is a homogeneous Poisson process with intensity $\lambda\in(0,\infty)$ that satisfies 
$N(t)\sim{\rm Poi}(\lambda t)$, $t\in[0,\infty),$
i.e., {the number of loss events up to time $t$} follows a Poisson distribution with parameter $\lambda t$. In this case, we have
$
\frac{N(t)}{t}\overset{a.s.}{\longrightarrow} \lambda$,
and, in particular,
$N(t)\overset{a.s.}{\longrightarrow} \infty$ as $t\to\infty$.

\color{black}
b) More generally, we may consider Cox processes. To construct these, we let $T$ be a random time change, i.e., an {adapted} stochastic process on $[0,\infty )$ with continuous and strictly increasing sample paths satisfying $T(0)=0$,  $\E(T(t)^2)<\infty$ and $\lim_{t\to\infty}T(t)=\infty$.  Let $\widetilde N$ be an independent homogeneous Poisson process with rate $\lambda\in(0,\infty)$. The Cox process is then given by $N(t)=\widetilde N(T(t))$, $t\in[0,\infty)$. {In the filtration $\mathcal F_t=\sigma(T(u):u\ge0)\vee\sigma(N(s):0\le s\le t)$, the compensator of $N$ is $\lambda T$.} See  \cite{gaku08} for details. The special case of an inhomogeneous Poisson process corresponds to the situation where $T$ is deterministic and absolutely continuous with respect to Lebesgue measure.
\end{exa}

 \color{black}
\subsection{Time-dependent interactions}\label{ext}

 For the  loss event counter $k\in \N$, we introduce tuples $(G_k,I_k,Z_k)$, extending equation \eqref{g:def}. To fix the notation, we adopt the convention that the loss event counter $k$ is the last subscript, i.e., $G_{i,j,k}=I_{i,j,k} \cdot Z_{i,j,k}$, $i, j \in \N$. We suppose that  $(G_k,I_k,Z_k)$, $k\in \N$, {each have the same distribution as} $(G,I,Z)$. {As shown} in Section \ref{exc} for a  jointly exchangeable array $G=(G_{i,j})_{ (i,j)\in\N\times \N}$, we have i.i.d. random variables  $\xi$, $\xi_{i}$, $\xi_{i,j}$, $i\leq j$, all uniformly distributed in $[0,1]$, along with a measurable function $h:[0,1]^4 \to \mathbb R$, such that almost surely
$$G_{i,j} = h(\xi,\xi_{i},\xi_{j},\xi_{i,j}), \quad (i,j) \in \N \times \N,$$
where we set $\xi_{j,i} : = \xi_{i,j}$ for $i\leq j$.  We do not assume that the arrays of random variables $G_k=(G_{i,j,k})_{ (i,j)\in\N\times \N}$, $k\in\N$, are  independent. Taking time-dependent interactions into account, we suppose the following more general  structure of the arrays $G_k=(G_{i,j,k})_{ (i,j)\in\N\times \N}$, $k\in\N$.

\begin{assumption}\label{a3}
Suppose we have   i.i.d. random variables  $\xi$,  $\xi_k$, $\xi_{i,k}$, $\xi_{i,j,k}$, $i\leq j$, $k\in\N$, all uniformly distributed in $[0,1]$, along with a measurable function $h:[0,1]^5 \to \mathbb R$, such that almost surely
$$G_{i,j,k} = h(\xi, \xi_k,\xi_{i,k}, \xi_{j,k}, \xi_{i,j,k}), \quad (i,j,k) \in \N \times \N \times \N,$$
where we set $\xi_{j,i,k} : = \xi_{i,j,k}$ for $i\leq j$.
\end{assumption}

{Under the stated assumptions, each array $G_k$ is jointly exchangeable.} Furthermore, the arrays of random variables $G_k=(G_{i,j,k})_{ (i,j)\in\N\times \N}$, $k\in\N$, have indeed all the same distribution. Moreover,  the arrays of random variables $G_k=(G_{i,j,k})_{ (i,j)\in\N\times \N}$, $k\in\N$, are not independent in general. We state the following remark on the joint distribution of these arrays.

\begin{rmk}
Under Assumption \ref{a3}, the array of random variables $(G_{i,j,k})_{ (i,j,k)\in\N\times \N\times \N}$ is  partially exchangeable in the sense that it has  the same distribution as $(G_{\pi(i),\pi(j),\rho(k)})_{(i,j,k) \in\N\times \N\times \N}$ for all bijective $\pi:\N\to\N$ and all bijective $\rho:\N\to\N${. This symmetry lies between joint and separate exchangeability. See \citet{ald0} and \citet{lee}.}
\end{rmk}

Now, let us go back to our counting process $N=(N(t))_{t\in [0,\infty)}$ introduced in Section~\ref{number}.  {As shown} there,  the counting process can be represented by the time points of  loss events $T_k$, $k\in\N,$ and there are measurable maps $\ell_1:[0,1]\rightarrow  [0,\infty]$ and $\ell_k:[0,1]\times [0,\infty]^{k-1}\rightarrow  [0,\infty]$, $k\ge 2$, and  i.i.d. random variables $\zeta_k$, $k\in\N$, all uniformly distributed in $[0,1]$,  such that almost surely
$$T_1=\ell_1(\zeta_1),~T_k=\ell_k(\zeta_k,T_1,\dots,T_{k-1}),~k\ge 2.$$
Taking  time-dependent interactions into account, we suppose the following more general structure on the counting process.

\begin{assumption}\label{a4}
{In addition to the variables in Assumption~\ref{a3}, we assume there are random variables} $\zeta_k$, $k\in\N$, all uniformly distributed in $[0,1]$,  such that  $\zeta_k$, $\xi$, $\xi_k$, $\xi_{i,k}$, $\xi_{i,j,k}$, $i\leq j$, $k\in\N$, {are independent. We also assume that there are measurable functions} $\ell_1:[0,1]^2\rightarrow  [0,\infty]$ and $\ell_k:[0,1]^2\times [0,\infty]^{k-1}\rightarrow  [0,\infty]$, $k\ge 2$, such that almost surely
$$T_1=\ell_1(\zeta_1,\xi),~T_k=\ell_k(\zeta_k,\xi,T_1,\dots,T_{k-1}),~k\ge 2.$$
\end{assumption}

Note that  the arrays of random variables $G_k=(G_{i,j,k})_{ (i,j)\in\N\times \N}$, $k\in\N$, and the counting process $N=(N(t))_{t\in [0,\infty)}$ are not independent in general. {As before}, the joint distribution of the  arrays and the counting process  can be decomposed as a mixture and we can consider the related  conditional joint distribution given  $\xi$. This conditional joint distribution  given $\xi$ has the property that  $N$, $G_k$, $k\in \N$, are  independent. {Conditional on $\xi$ and $(\xi_k)_{k\in\N}$, each array $G_k$ is dissociated.} The deliberations in Section \ref{rmkdis} apply correspondingly.

The random effects that occur in our model have the following interpretation. For $i\leq j$, $k\in\N$, the  random variable  $\xi_{k}$  covers systemic background shocks with simultaneous influence on all   infections and transmissions  in  the network {at loss event $k$. The random variable} $\xi_{i,k}$  has influence on the external infection of the $i$-th entity {at loss event $k$. The random variable}  $\xi_{i,j,k}$ has influence on the internal transmission from entity $j$ to entity $i$  in  the network at the point in time of the $k$-th loss event and the  random variables  $\zeta_{1},\dots,\zeta_{k}$ have influence on  the point in time of the  $k$-th loss event. The  random variable $\xi$ is a global effect  with simultaneous  influence on all internal transmissions and all  external infections in  the network as well as  on all points in time of the  loss events; it covers in particular time-dependent interactions in  the network and interactions between the network structure and the counting process of loss events.

\color{black}

\subsection{The  loss of the insurance company}\label{total}

We denote the loss incurred by entity $i$ up to time $t$, caused either by transmissions from entity $j$ (if $i\neq j$) or by infections (if $i= j$),  by 
\begin{align*}
X_{i,j}(t)=\sum_{k=1}^{N(t)}G_{i,j,k}.
\end{align*} 
{For an insurance portfolio of entities indexed by} $\{1,2, \dots, n\}$, where $n\in \N$, the total loss experienced by entity $i=1,\dots,n$ up to time $t$ is
\begin{align*}
Y_i(n,t)=\sum_{j=1}^n X_{i,j}(t).
\end{align*} 

Finally, the total loss of the insurance company up to time $t$  is given by the sum of the individual losses in the insurance portfolio:
\begin{align*}
S(n,t)=\sum_{i=1}^n  Y_i(n,t).
\end{align*}
Alternatively, this can be expressed as 
\begin{align*}
S(n,t)=\sum_{k=1}^{N(t)}L_{k}(n),
\end{align*}
where  $L_{k}(n)=\sum_{i,j=1}^n G_{i,j,k}$ denotes the total loss incurred during loss event $k\in\N$.

\color{black}

\section{Limit results}\label{theory}

In this section, we derive several limit theorems for large portfolios, long-term horizons, and their combination. To be more specific, Section~\ref{n incr} establishes  central limit theorems for the total loss $S(n,t)$ {as the number of insurance contracts $n$ tends to infinity with $t$ fixed. It also extends the results} to large-portfolio limits over long time horizons. Finally, in Section~\ref{spec}, we discuss additionally central limit theorems for special cases and related Berry-Esseen type results. In what follows, {we denote by} $\Phi_{\mu,\sigma^2}$  the distribution function of $N(\mu,\sigma^2)$, $\mu \in \mathbb{R}$, $\sigma^2 \in [0,\infty)$, {with the convention that $N(\mu,0)$ is the point mass at $\mu$ and $\Phi_{\mu,0}$ is its distribution function.} Moreover, we set $0/0=0$.

\subsection{Central limit theorems}\label{n incr}

Fixing the time horizon $t$, we first establish a central limit theorem for the large portfolio limit. We use a  result in \citet{sil}  for jointly exchangeable and dissociated arrays, leading to the following statement.

\begin{thm}\label{thm2}
Suppose that Assumptions \ref{a1}--\ref{a4} hold. Setting
\begin{align*}
\sigma^2(t)=\Cov(X_{1,2}(t)+X_{2,1}(t),X_{1,3}(t)+X_{3,1}(t)|\xi, \xi_{k},\zeta_{k},k\in\N),
\end{align*} 
we suppose that $\sigma^2(t)>0$ almost surely  if $N(t)>0$.  {With}
\begin{align*}
\mu(n,t)=\E( S(n,t)|\xi, \xi_{k},\zeta_{k},k\in\N),
\end{align*}
{we obtain} 
\begin{align*}
\frac{ S(n,t)-\mu(n,t)}{n^{3/2}} \overset{d}{\longrightarrow}\tau ~\text{as}~n\to\infty,
\end{align*}
where $\tau$ is a real-valued random variable with  distribution function 
$$P(\tau\le s)=\E(\Phi_{0,\sigma^2(t)}(s)),~s\in\R.$$
\end{thm}

A modified version of Theorem~\ref{thm2} also holds in the case where time increases simultaneously with the number of contracts. The simultaneous limit introduces additional challenges.  Nevertheless, in our setting, we show that these issues do not arise, thereby establishing the validity of the joint limit. We will verify Lindeberg's condition by using again  a limit result in \citet{sil} for jointly exchangeable and dissociated arrays and subsequently apply the central limit theorem of Lindeberg-Feller. 

\begin{thm}\label{thm3}
Suppose that Assumptions \ref{a1}--\ref{a4} hold. Setting
\begin{align*}
\mu(n,t)=\E( S(n,t)|\xi, \xi_{k},\zeta_{k},k\in\N)
\end{align*}
and
\begin{align*}
\sigma^2(n,t)=\Var( S(n,t)-\mu(n,t)|\xi, \zeta_{k},k\in\N),
\end{align*}
we assume that   $\sigma^2(n,t)>0$  almost surely  if $N(t)>0$. Setting
\begin{align*}
\sigma^2=\Cov(G_{1,2,1}+G_{2,1,1},G_{1,3,1}+G_{3,1,1}|\xi, \xi_{1}),
\end{align*} 
we suppose that $\sigma^2>0$ almost surely. Assume in addition $N(t) \overset{a.s.}{\longrightarrow} \infty$ as $t \to \infty$. Then, we have for the total portfolio loss $S(n,t)$ the convergence in distribution 
\begin{align*}
\frac{ S(n,t)-\mu(n,t)}{\sigma(n,t)} \overset{d}{\longrightarrow}N(0,1) ~\text{as}~\min(n,t)\to\infty.
\end{align*}
\end{thm}

We now consider a central limit theorem for the total loss of the insurance company as the time horizon $t \to \infty$ increases, while the number of insurance contracts $n \in \mathbb{N}$ remains fixed. 

\begin{thm}\label{thm4}
Suppose that Assumptions \ref{a1}--\ref{a4} hold. Setting
\begin{align*}
\mu(n,t)=\E( S(n,t)|\xi, \zeta_{k},k\in\N)
\end{align*}
and
\begin{align*}
\sigma^2(n,t)=\Var( S(n,t)-\mu(n,t)|\xi, \zeta_{k},k\in\N),
\end{align*}
we assume that   $\sigma^2(n,t)>0$  almost surely  if $N(t)>0$.  Assume in addition $N(t) \overset{a.s.}{\longrightarrow} \infty$ as $t \to \infty$. Then, we have for the total portfolio loss $S(n,t)$ the convergence in distribution 
\begin{align*}
\frac{ S(n,t)-\mu(n,t)}{\sigma(n,t)} \overset{d}{\longrightarrow}N(0,1) ~\text{as}~t\to\infty.
\end{align*}

\end{thm}

\subsection{Limit theorems in a special case}\label{spec}

{Throughout this subsection, we impose Assumptions~\ref{a1}--\ref{a4} and assume that}
$$G_{i,j,k} = h(\xi_{i,k}, \xi_{j,k}, \xi_{i,j,k}), \quad (i,j,k) \in \N \times \N \times \N,$$
almost surely and that
$$T_1=\ell_1(\zeta_1),~T_k=\ell_k(\zeta_k,T_1,\dots,T_{k-1}),~k\ge 2,$$
almost surely.    Now,  the arrays of random variables $G_k$, $k\in \N$, {are each dissociated, and $N,G_1,G_2,\dots$ are mutually independent.} This setting covers in particular our examples in Section \ref{exa}. With 
\begin{align*}
\mu_L(n)=\E(L_{1}(n))=n\E(G_{1,1})+n(n-1)\E(G_{1,2}),
\end{align*}
and
\begin{align*}
\sigma^2_L(n)=&\Var(L_{1}(n))\\
=&n\Var(G_{1,1})\\
+&n(n-1)\\
\times&\big(2\Cov(G_{1,1},G_{1,2})+2\Cov(G_{1,1},G_{2,1})+\Var(G_{1,2})+\Cov(G_{1,2},G_{2,1})\big)\\
+&n(n-1)(n-2)\\
\times&\big(\Cov(G_{1,2},  G_{1,3})+\Cov(G_{1,2},  G_{3,1} )+\Cov(G_{2,1},  G_{1,3})+\Cov(G_{2,1}, G_{3,1} )\big)
\end{align*}
as the expectation and variance of $L_{1}(n)$, we obtain $\mu(n,t)=N(t)\mu_L(n)$ and $\sigma^2(n,t)=N(t)\sigma_L^2(n)$ in our previous notations.  Under the  assumptions of Theorem \ref{thm2}, it is 
$$\Cov(X_{1,2}(t)+X_{2,1}(t),X_{1,3}(t)+X_{3,1}(t)|\zeta_{k},k\in\N)=N(t)\Cov(G_{1,2}+G_{2,1},G_{1,3}+G_{3,1})$$ as well as 
$$\lim_{n\to\infty}\frac{n^{3/2}}{\sigma_L(n)}=\frac{1}{\sqrt{\Cov(G_{1,2}+G_{2,1},G_{1,3}+G_{3,1})}}$$
and  Theorem \ref{thm2} implies
\begin{align*}
\frac{ S(n,t)-N(t)\mu_L(n)}{\sigma_L(n)} \overset{d}{\longrightarrow}\tau ~\text{as}~n\to\infty,
\end{align*}
with   a real-valued random variable $\tau$ with distribution function $P(\tau\le s)=\sum_{m=0}^\infty P(N(t)=m)\Phi_{0,m}(s)$ for $s \in \mathbb{R}$. Under the  assumption of  Theorem \ref{thm3} or  Theorem \ref{thm4}, it holds that
\begin{align*}
\frac{ S(n,t)-N(t)\mu_L(n)}{\sqrt{N(t)}\sigma_L(n)} \overset{d}{\longrightarrow} N(0,1) ~ \text{as}~ \min(n, t) \to \infty~ \text{or}~ t \to \infty.
\end{align*}
Motivated by the previous central limit theorems, we now derive related Berry-Esseen type results. We use findings in  \citet{zha} for generalized $U$-statistics and the usual Berry-Esseen theorem. Part a)  of the following result is tailored for a fixed point in time $t$ and an increasing  number of insurance contracts $n$ (see Theorem \ref{thm2}), while part b) could be useful in situations where the point in time $t$ tends to infinity while the number of insurance contracts $n$ is either constant or tends to infinity {simultaneously} (see Theorem \ref{thm3} or  Theorem \ref{thm4}). We  use the conventions $0/0=0$, $x/0=\infty$ for $x\in\R\setminus \{0\}$ and $0\cdot\infty=0$.

\begin{thm}\label{thm40}
a) Under the assumptions of Theorem \ref{thm2}, it is
\begin{align*}
\sup_{s\in\R}| P(S(n,t)\le s)-\E( \Phi_{N(t)\mu_L(n),N(t)\sigma_L^2(n)}(s)) |\le\frac{\eta^2P(N(t)>0)}{\sqrt n}+   o\bigg(\frac{1}{\sqrt{n}}\bigg)~\text{as}~n\to\infty,
\end{align*}
where 
$$\eta^2=\frac{16(\E((G_{1,2}+G_{2,1}-\E(G_{1,2}+G_{2,1}))^4){\color{black})}^{1/2}}{\Cov( G_{1,2}+G_{2,1}, G_{1,3}+G_{3,1})}.$$

b) Under the assumptions of Theorem \ref{thm3} or  Theorem \ref{thm4}, we have
\begin{align*}
\sup_{s\in\R}|P(S(n,t)\le s)- \E(\Phi_{N(t)\mu_L(n),N(t)\sigma_L^2(n)}(s))|\le \E\bigg(\frac{I(N(t)>0)}{\sqrt{N(t)}}\bigg)\eta^2(n),
\end{align*}
where 
$$\eta^2(n)=\frac{C_{\operatorname{BE}}\E(|L_1(n)-\E(L_1(n))|^3)}{\Var(L_1(n)-\E(L_1(n)))^{3/2}}$$ 
with $C_{\operatorname{BE}}$ as the usual Berry-Esseen constant.
\end{thm}

{For part b) of Theorem~\ref{thm40}, if $N(t)\to\infty$ almost surely, then $\E({I(N(t)>0)}/{\sqrt{N(t)}})\to 0$ by dominated convergence.} We state the following remark on the application of the previous results.

\begin{rmk}\label{rmk2}
Our previous results {suggest approximating} the distribution of the total loss $S(n,t)$ by a mixture of normal {distributions. If the assumptions of Theorem} \ref{thm2}, Theorem \ref{thm3} or  Theorem \ref{thm4} are valid, we have the approximation
\begin{align*}
P\big (S(n,t)\leq s\big )\approx \sum_{m=0}^\infty P(N(t)=m)\Phi_{m\mu_L(n),m\sigma_L^2(n)}(s)
\end{align*}
for sufficiently large $n$, $\min(n,t)$ or $t$.
\end{rmk}
Additionally, we now study special cases related to Poisson counting processes. Applying Wald's equations, we obtain the expected total loss by 
\begin{align*}
\mu_S(n,t)=\E\big(S(n,t)\big)=\E(N(t)) \mu_L(n)
\end{align*}
and the variance of the total loss by  
\begin{align*}
\sigma^2_S(n,t)=\Var\big(S(n,t)\big)=\Var(N(t))\mu_L(n)^2+\E(N(t))\sigma^2_L(n).
\end{align*}
First, we consider a homogeneous Poisson process with intensity $\lambda\in(0,\infty)$, exploiting the fact that Poisson processes are L\'evy processes, meaning that the increments $N(\lfloor t \rfloor) - N(\lfloor t \rfloor - 1), \dots, N(1) - N(0)$ are independent and identically distributed. This property enables a particularly convenient approximation of the insurance company's loss distribution in applications; see Remark~\ref{rmk3b}. Note that it is $\E(N(t)) =\Var(N(t))=\lambda t$ for $t\ge 0$ in this situation. Second, we extend the result to Cox processes, which are defined in terms of increasing random time changes $T$; see Example~\ref{exatime} b).

\begin{thm}\label{thm5}
Suppose that Assumptions \ref{a1}--\ref{a4} hold. 

a) If $\sigma^2_S(n,t)>0$ for sufficiently large $t$ and the counting process $N=(N(t))_{t\in[0,\infty)}$ is a homogeneous Poisson process with intensity $\lambda\in(0,\infty)$, then we have convergence in distribution
\begin{align*}
\frac{S(n, t)-\lambda t \mu_L(n)}{\sqrt{\lambda t(\mu_L(n)^2+\sigma^2_L(n))}}\overset{d}{\longrightarrow}N(0,1)~\text{as}~t\to\infty.
\end{align*}

b) Suppose that $\sigma^2_S(n,t)>0$ for sufficiently large $t$ and that the jump process $N=(N(t))_{t\in[0,\infty)}$ is of the form $N(t)=\widetilde N(T(t))$, $t\in[0,\infty)$, where  $\widetilde N=(\widetilde N(t))_{t\in[0,\infty)}$ is a homogeneous Poisson process with intensity $\lambda\in(0,\infty)$, and $T=(T(t))_{t\in[0,\infty)}$ is a stochastic process with values in $[0,\infty)$ satisfying $\lim_{t\to\infty}T(t)=\infty$ almost surely. Assume further that the random elements $\widetilde N$, $T$, and $G_{k}$, $k\in\N$, are  independent. Then we have the convergence in distribution
\begin{align*}
\frac{ S(n,t) -\lambda T(t) \mu_L(n)}{\sqrt{\lambda T(t) (\mu_L(n)^2+\sigma^2_L(n))}}\overset{d}{\longrightarrow}N(0,1)~\text{as}~t\to\infty.
\end{align*}
\end{thm}

We state the following application-oriented remark.

\begin{rmk}\label{rmk3b}
{As in Remark~\ref{rmk2}, the previous results suggest approximating} the distribution of the insurance company's loss $S(n,t)$ for sufficiently large $t$.  Suppose that the hypotheses of Theorem~\ref{thm5} hold and consider the setting of part a).  Letting $s \in \R$, we use  for sufficiently large $t$ the approximation
\begin{align*}
P\big(S(n,t) \le s\big)\approx \Phi_{\lambda t \mu_L(n),\lambda t(\mu_L(n)^2+\sigma^2_L(n))}(s).
\end{align*}
Suppose that the hypotheses of Theorem~\ref{thm5}  hold and consider the setting of part b). Denoting by $\mu_{T(t)}$ the distribution of $T(t)$ and letting $s \in \R$, we use  for sufficiently large $t$ the approximation
\begin{align*}
P\big (S(n,t)  \le s\big)\; \approx \; \int \Phi_{\lambda u \mu_L(n),\lambda u(\mu_L(n)^2+\sigma^2_L(n))}(s) \,d\mu_{T(t)}(u).
\end{align*}
\end{rmk}

 \color{black}

\section{Relation to  the  collective risk model}\label{rel}

Before discussing numerical case studies, another important aspect deserves emphasis - namely, that under certain conditions, the proposed model can be reformulated as a classical frequency-severity model, also known as the collective model. These conditions require that the arrays $I$ and $Z$ are independent and that the conditional loss expenditures $Z_{i,j}$, $i,j \in \mathbb{N}$, are independent, as discussed in Example~\ref{exa2} a). To align with the notation of a classical collective model, the frequency component must be adjusted, while the severity component follows directly from the conditional loss expenditures. In this section, we briefly discuss this transformation.

 \color{black}
We suppose that Assumptions \ref{a1} - \ref{a4} are valid in the special case  introduced at the beginning of Section \ref{spec}.  Moreover, we suppose that the conditional loss expenditures $Z_{i,j,k}$, $i,j,k \in \mathbb{N}$, are independent and identically distributed and {that the indicator array $(I_{i,j,k})_{i,j,k\in\N}$, the expenditure array $(Z_{i,j,k})_{i,j,k\in\N}$, and the counting process $(N(t))_{t\in[0,\infty)}$ are mutually independent}. Let $\widetilde{Z}_i$, $i \in \mathbb{N}$, denote another independent sequence of i.i.d. random variables such that $\widetilde{Z}_1 \overset{d}{=} Z_{1,1,1}$. The total number of claims up to time $t$ is given by
\begin{align*}
M(n,t) = \sum_{k=1}^{N(t)} \sum_{i=1}^n \sum_{j=1}^n I_{i,j,k}.
\end{align*}
Defining $M(n) = (M(n,t))_{t \in [0,\infty)}$, the random elements $M(n)$ and $\widetilde{Z}_i$, $i \in \mathbb{N}$, are jointly independent.

\begin{thm} \label{thmrel}
Under the assumptions outlined above, we obtain equality in distribution with a collective model:
\begin{align*}
\big(S(n, t),M(n, t)\big) \overset{d}{=} \left(\sum_{i=1}^{M(n,t)} \widetilde Z_i, M(n,t)\right).
\end{align*}
\end{thm}

\begin{rmk}\label{rmk4}
In the absence of transmissions and under independent infections, as described in Example~\ref{exa1} a), the results of Theorem~\ref{thmrel} simplify further. In this case, the total number of claims is given by
$M(n,t) = \sum_{k=1}^{N(t)} \sum_{i=1}^n I_{i,i,k}.$
We highlight two special cases:
a) If $N(t) \equiv 1$, then the claim count follows a binomial distribution:
$M(n,t) \sim {\rm Bin}(n, p_J)$.
b) If $N$ is a homogeneous Poisson process with intensity $\lambda \in (0,\infty)$, and if $n = 1$ and $p_J = 1$, then the claim count follows a Poisson distribution:
$M(n,t) \sim {\rm Poi}(\lambda t)$.
\end{rmk}

\section{Approximations and interaction effects -- case studies}\label{impsim}

In this section, we present numerical experiments implemented in the statistical software R. Section~\ref{imp} outlines the general setup of the case studies and provides moment computations required for the approximations derived in Section~\ref{theory}. In Section~\ref{sim}, we assess the quality of selected distributional approximations established in Section~\ref{theory}.  Finally, Section~\ref{comp} examines in greater detail the impact of dependencies in loss occurrences and conditional loss {expenditures} on the distribution of total losses, with a particular focus on their tails.

\subsection{Implementation of the simulation studies}\label{imp}
 
\color{black}
{In the simulations, we impose Assumptions~\ref{a1}--\ref{a4} and use the special case described at the start of Section~\ref{spec}.} Infections and transmissions $I$ will be modeled using the Erd\H{o}s-R{\'e}nyi model with infections (Example~\ref{exa1} d)) and the contagion model (Example~\ref{exa1} e)). We focus on positively dependent conditional loss expenditures $Z$ (Example~\ref{exa2} c)). 
To simplify notation, we do not assume that the independent shocks are uniformly distributed on $[0,1]$, {but instead assume that they belong to a parametric family of distributions}. Specifically, we set
\begin{align}\label{lossexa}
Z_{i,j,k} =\widetilde  Z_{j,k} + \varepsilon_{i,j,k}, \quad i,j,k \in \mathbb{N},
\end{align}
where the contagious components $\widetilde Z_{j,k}$, $j,k \in \mathbb{N}$, are independent, $[0,\infty)$-valued random variables, each following a gamma distribution $\rm{G}(\nu, \kappa)$ with a common shape parameter $\nu \in (0,\infty)$ and scale parameter $\kappa \in (0,\infty)$. The gamma distribution forms a rich family that includes the exponential, chi-squared, and Erlang distributions. The shocks $\varepsilon_{i,j,k}$, $i,j,k \in \mathbb{N}$, are independent random variables following a half-normal distribution $\rm{HN}(\sigma^2)$ with parameter $\sigma^2 \in (0,\infty)$. {The indicator array $(I_k)_{k\in\N}$, the family of gamma variables, and the family of half-normal shocks are mutually independent.} The counting process $N$ is taken to be a Poisson process with intensity $\lambda \in (0,\infty)$.  We explicitly compute the relevant moments for the case studies:
\begin{itemize}
\item[(i)] {\bf Erd\H{o}s-R{\'e}nyi model with infections} (Example~\ref{exa1} d)) 
\begin{align*}
 \mu_L(n) \; = \; \; &np_J\bigg(\nu\kappa+\sqrt{\frac{2\sigma^2}{\pi}}\bigg)+n(n-1)p_K\bigg(\nu\kappa+\sqrt{\frac{2\sigma^2}{\pi}}\bigg),\\
\sigma^2_L(n)\; =\;\;  &n\Bigg(p_J\bigg(\nu\kappa^2+\nu^2\kappa^2+2\nu\kappa\sqrt{\frac{2\sigma^2}{\pi}}+\sigma^2\bigg)-p_J^2\bigg(\nu\kappa+\sqrt{\frac{2\sigma^2}{\pi}}\bigg)^2\Bigg)\\
&+n(n-1)\Bigg(p_K\bigg(\nu\kappa^2+\nu^2\kappa^2+2\nu\kappa\sqrt{\frac{2\sigma^2}{\pi}}+\sigma^2\bigg)-p_K^2\bigg(\nu\kappa+\sqrt{\frac{2\sigma^2}{\pi}}\bigg)^2\\
&+2p_Jp_K\nu\kappa^2+(p_K-p_K^2)\bigg(\nu\kappa+\sqrt{\frac{2\sigma^2}{\pi}}\bigg)^2\Bigg)\\
&+n(n-1)(n-2)p_K^2\nu\kappa^2.
\end{align*}
\item [(ii)] {\bf Contagion model} (Example~\ref{exa1} e))
\begin{align*}
 \mu_L(n)\;= \; \; & np_J\bigg(\nu\kappa+\sqrt{\frac{2\sigma^2}{\pi}}\bigg)+n(n-1)p_Jp_K\bigg(\nu\kappa+\sqrt{\frac{2\sigma^2}{\pi}}\bigg),\\
\sigma^2_L(n)\;=\; \; &n\Bigg(p_J\bigg(\nu\kappa^2+\nu^2\kappa^2+2\nu\kappa\sqrt{\frac{2\sigma^2}{\pi}}+\sigma^2\bigg)-p_J^2\bigg(\nu\kappa+\sqrt{\frac{2\sigma^2}{\pi}}\bigg)^2\Bigg)\\
&+n(n-1)\Bigg(p_Jp_K\bigg(\nu\kappa^2+\nu^2\kappa^2+2\nu\kappa\sqrt{\frac{2\sigma^2}{\pi}}+\sigma^2\bigg)-p_J^2p_K^2\bigg(\nu\kappa+\sqrt{\frac{2\sigma^2}{\pi}}\bigg)^2\\
&+2p_Jp_K\bigg(\nu\kappa^2+\nu^2\kappa^2+2\nu\kappa\sqrt{\frac{2\sigma^2}{\pi}}+\frac{2\sigma^2}{\pi}\bigg)-2p_J^2p_K\bigg(\nu\kappa+\sqrt{\frac{2\sigma^2}{\pi}}\bigg)^2\Bigg)\\
&+n(n-1)(n-2)\Bigg(p_Jp_K^2\bigg(\nu\kappa^2+\nu^2\kappa^2+2\nu\kappa\sqrt{\frac{2\sigma^2}{\pi}}+\frac{2\sigma^2}{\pi}\bigg)-p_J^2p_K^2\bigg(\nu\kappa+\sqrt{\frac{2\sigma^2}{\pi}}\bigg)^2\Bigg).
\end{align*}
\end{itemize}
In our simulations, we fix the model parameters as follows: $p_J = 0.25$, $p_K = 0.35$, $\nu = 0.75$, $\kappa = 1.75$, $\sigma = 1.5$, and $\lambda = 1$.

\subsection{Accuracy of the approximations}\label{sim}

We investigate the quality of selected distributional approximations suggested by our findings in Section~\ref{theory} through simulation. We consider the model specifications described in Section~\ref{imp}. 

Theorem~\ref{thm2} addresses the case where the number of insurance contracts increases. Here, we fix the time point at $t = 1$ and vary the number of insurance contracts as $n = 15, 50, 200$. Theorem~\ref{thm3} establishes asymptotics when both the number of insurance contracts and the time horizon increase jointly. We examine the approximations for the pairs $(n,t) = (15,15), (50,25), (200,80)$. Finally, we assess the quality of the approximation in Theorem~\ref{thm5} a), where the number of contracts remains fixed while the time horizon increases. We set the number of insurance contracts to $n = 10$ and vary the time points as $t = 15, 25, 80$. To obtain the distributions under consideration, we perform Monte Carlo simulations with 100,000 replications for each scenario.

Figure~\ref{fig3} presents P-P and Q-Q plots for the distributions in Theorem~\ref{thm2}, where the Monte Carlo estimate of the finite-sample distribution is plotted on the abscissa against the asymptotic distribution on the ordinate. Overall, the accuracy of the finite-sample approximation improves as the number of insurance contracts $n$ increases, ultimately reaching a high level of precision, confirming the practical applicability of Theorem~\ref{thm2}. More specifically, for $n = 15$, both the P-P and Q-Q plots indicate that the approximation is not yet highly accurate. A likely contributor to deviations in the upper tail is the heavier tail of the gamma distribution used for modeling loss expenditures in the finite distribution compared to the normal distributions in the asymptotic case. For $n = 50$, the P-P plots already suggest a good approximation in the central part of the distribution, though small deviations remain in the tails. For $n = 200$, the Q-Q plots indicate a highly precise approximation across the entire distribution. These findings hold for both the Erd\H{o}s-R{\'e}nyi model with infections and the contagion model. The empty regions in the P-P plots result from the strictly positive probability mass at zero in the underlying distributions.

Figure~\ref{fig4} also presents P-P and Q-Q plots for approximations based on Theorem~\ref{thm3}, where the finite-sample distribution on the abscissa is plotted against the asymptotic distribution on the ordinate. The accuracy of the finite-sample approximation improves as both the number of insurance contracts $n$ and the time horizon $t$ increase, ultimately reaching a high level of precision, confirming the practical applicability of Theorem~\ref{thm3}. More specifically, for $(n,t) = (15,15)$, the P-P and Q-Q plots indicate small deviations in the approximation. We observe discrepancies in the upper tail of the distribution, likely due to the heavier tail of the gamma distribution compared to the asymptotic normal distribution. For $(n,t) = (50,25)$, the approximation improves significantly, though minor deviations remain in the tails. Finally, for $(n,t) = (200,80)$, remaining deviations essentially vanish. These results hold for both the Erd\H{o}s-R{\'e}nyi model with infections and the contagion model. Overall, the approximation appears to perform better than in the setting of Theorem~\ref{thm2} for a given portfolio size $n$. This improvement is due to the fact that in Theorem~\ref{thm3}, both $n$ and $t$ tend to infinity.

Figure~\ref{fig5} presents P-P and Q-Q plots for the distributions in Theorem~\ref{thm5}~a), where the empirical distribution on the abscissa is plotted against the asymptotic distribution on the ordinate. The accuracy of the approximation improves as the time horizon $t$ increases. A closer inspection of the P-P and Q-Q plots for $t = 15$ reveals that the approximation is not yet highly accurate. As before, deviations in the upper tail are likely due to the difference in tail heaviness between the gamma distribution used for modeling loss expenditures in the finite distribution and the asymptotic normal distribution. For $t = 25$, the P-P plots already indicate a good approximation in the central part of the distribution, while the Q-Q plots still show some deviations in the tails. By $t = 80$, these tail deviations are further reduced in the Q-Q plots. In general, {the approximation based on Theorem~\ref{thm3} appears more accurate than the one based on Theorem~\ref{thm5}(a)} for the same values of $t$. The likely reason is that in Theorem~\ref{thm3}, both $n$ and $t$ tend to infinity, whereas in Theorem~\ref{thm5} a), $n$ is fixed and only $t$ tends to infinity.

\begin{figure}[!t]
\hspace{0.75cm}
\begin{subfigure}{.275\textwidth}
\captionsetup{font=tiny}
\captionsetup{justification=centering}
 \centering
  \includegraphics[width=1\linewidth]{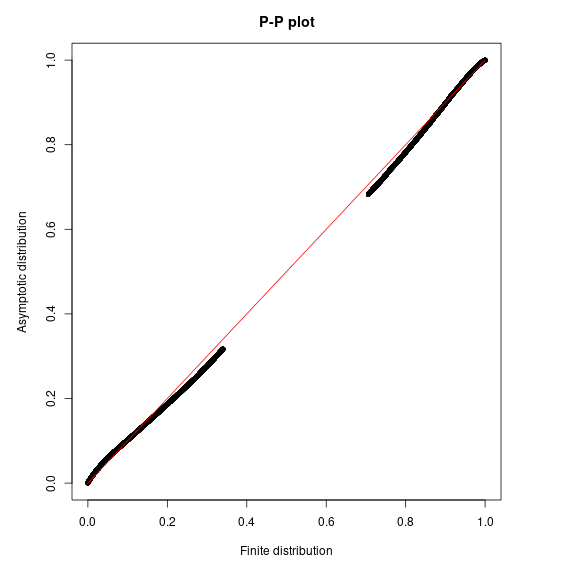}
    \caption*{Erd\H{o}s-R{\'e}nyi with infections, $n=15$.}
\end{subfigure}%
\begin{subfigure}{.275\textwidth}
\captionsetup{font=tiny}
\captionsetup{justification=centering}
 \centering
  \includegraphics[width=1\linewidth]{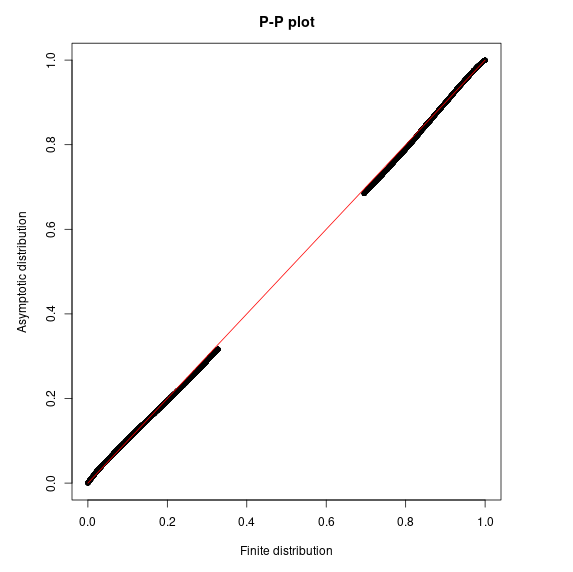}
    \caption*{Erd\H{o}s-R{\'e}nyi with infections, $n=50$.}
\end{subfigure}%
\begin{subfigure}{.275\textwidth}
\captionsetup{font=tiny}
\captionsetup{justification=centering}
 \centering
  \includegraphics[width=1\linewidth]{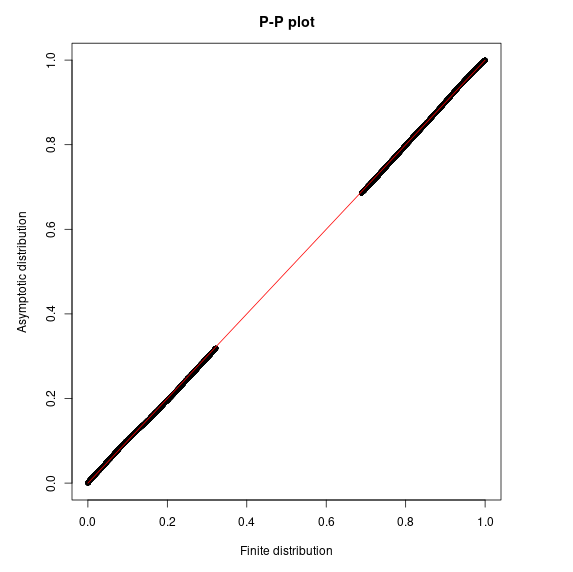}
    \caption*{Erd\H{o}s-R{\'e}nyi with infections, $n=200$.}
\end{subfigure}

\hspace{0.75cm}
\begin{subfigure}{.275\textwidth}
\captionsetup{font=tiny}
\captionsetup{justification=centering}
 \centering
  \includegraphics[width=1\linewidth]{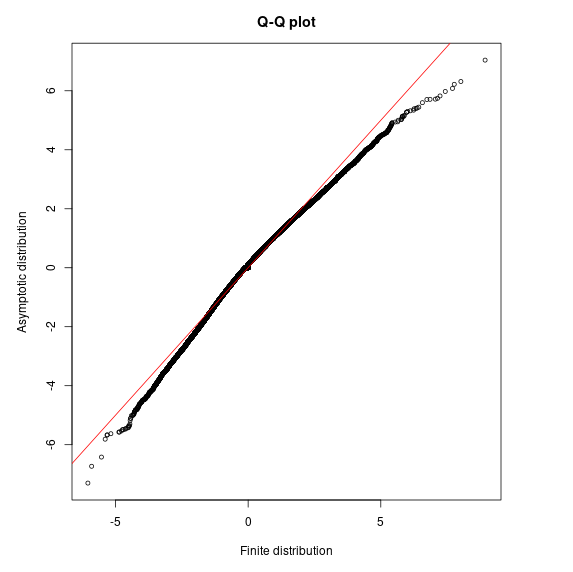}
    \caption*{Erd\H{o}s-R{\'e}nyi  with infections, $n=15$.}
\end{subfigure}%
\begin{subfigure}{.275\textwidth}
\captionsetup{font=tiny}
\captionsetup{justification=centering}
 \centering
  \includegraphics[width=1\linewidth]{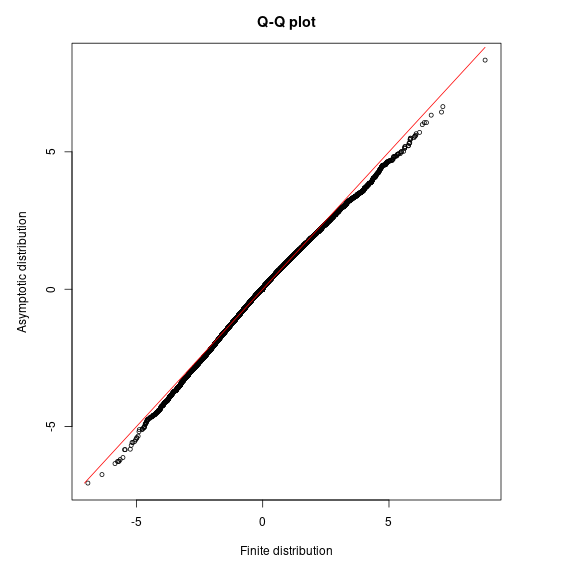}
    \caption*{Erd\H{o}s-R{\'e}nyi with infections, $n=50$.}
\end{subfigure}%
\begin{subfigure}{.275\textwidth}
\captionsetup{font=tiny}
\captionsetup{justification=centering}
 \centering
  \includegraphics[width=1\linewidth]{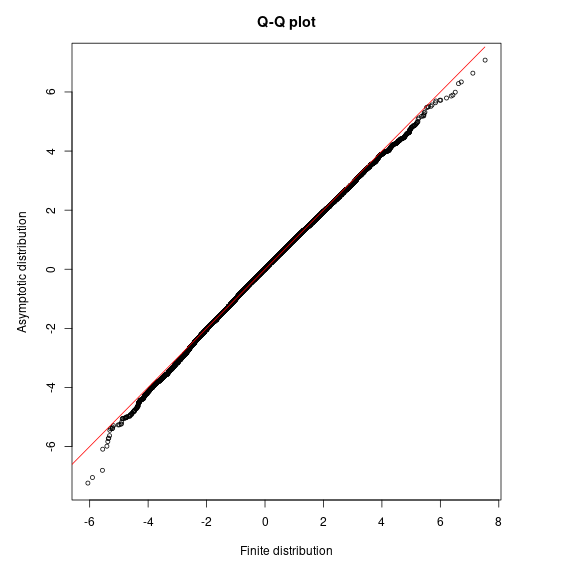}
    \caption*{Erd\H{o}s-R{\'e}nyi with infections, $n=200$.}
\end{subfigure}

\hspace{0.75cm}
\begin{subfigure}{.275\textwidth}
\captionsetup{font=tiny}
\captionsetup{justification=centering}
 \centering
  \includegraphics[width=1\linewidth]{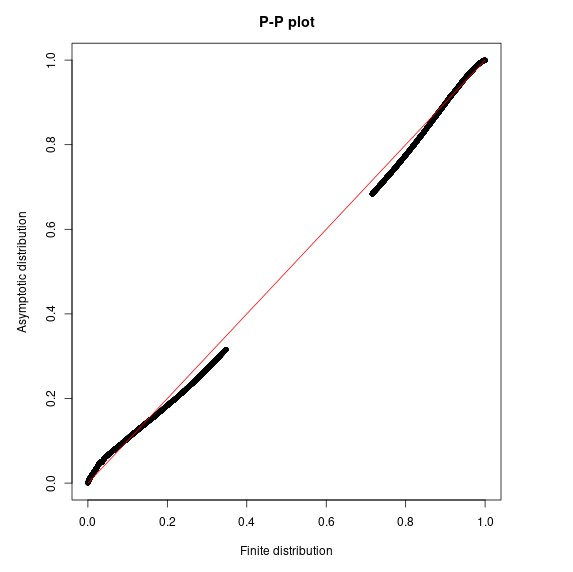}
    \caption*{Contagion case, $n=15$.}
\end{subfigure}%
\begin{subfigure}{.275\textwidth}
\captionsetup{font=tiny}
\captionsetup{justification=centering}
 \centering
  \includegraphics[width=1\linewidth]{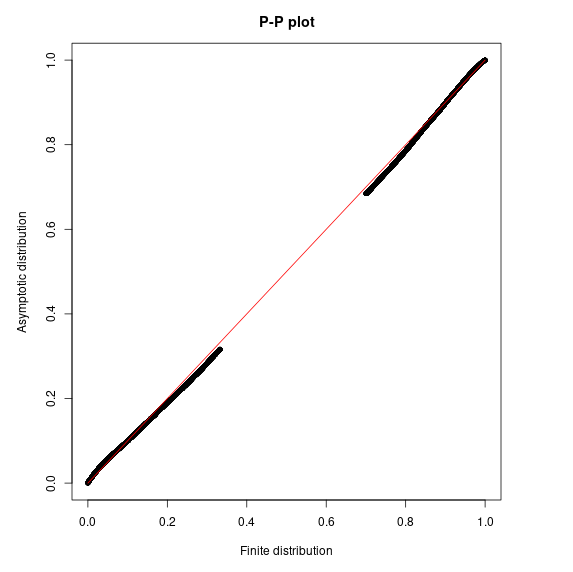}
    \caption*{Contagion case, $n=50$.}
\end{subfigure}%
\begin{subfigure}{.275\textwidth}
\captionsetup{font=tiny}
\captionsetup{justification=centering}
 \centering
  \includegraphics[width=1\linewidth]{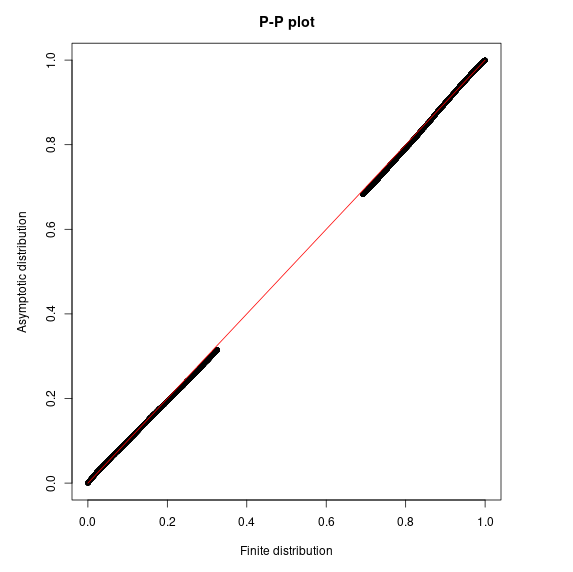}
    \caption*{Contagion case, $n=200$.}
\end{subfigure}

\hspace{0.75cm}
\begin{subfigure}{.275\textwidth}
\captionsetup{font=tiny}
\captionsetup{justification=centering}
 \centering
  \includegraphics[width=1\linewidth]{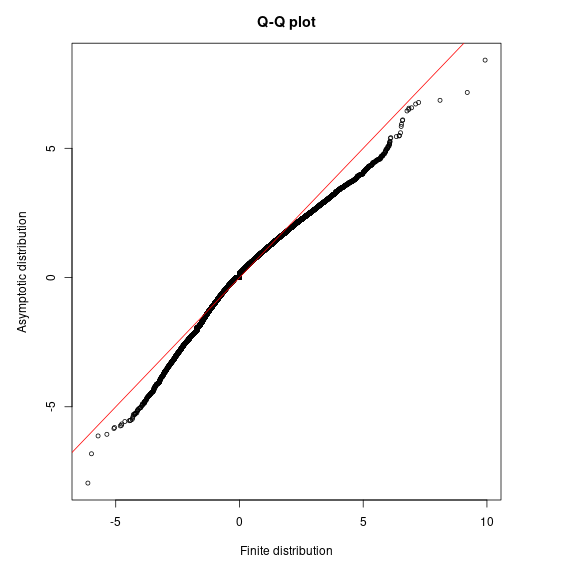}
    \caption*{Contagion case, $n=15$.}
\end{subfigure}%
\begin{subfigure}{.275\textwidth}
\captionsetup{font=tiny}
\captionsetup{justification=centering}
 \centering
  \includegraphics[width=1\linewidth]{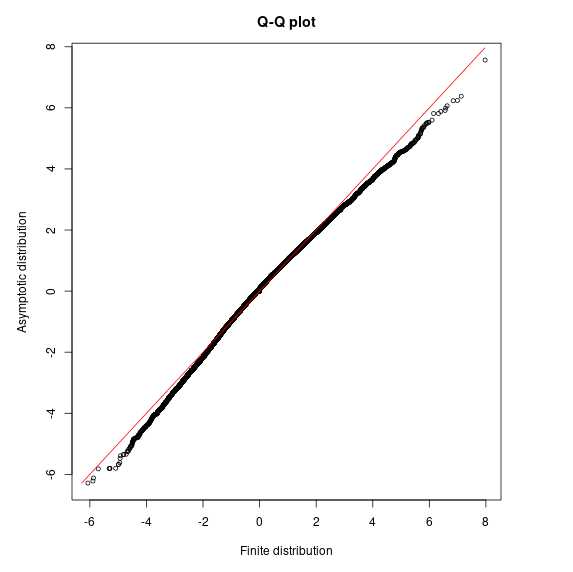}
    \caption*{Contagion case, $n=50$.}
\end{subfigure}%
\begin{subfigure}{.275\textwidth}
\captionsetup{font=tiny}
\captionsetup{justification=centering}
 \centering
  \includegraphics[width=1\linewidth]{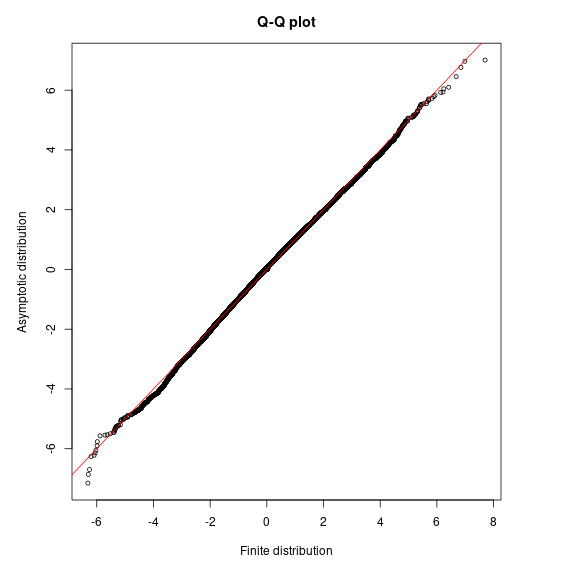}
    \caption*{Contagion case, $n=200$.}
\end{subfigure}
\caption{Plots related to Theorem \ref{thm2}.}
  \label{fig3}
\end{figure}

\begin{figure}[!t]
\hspace{0.75cm}
\begin{subfigure}{.275\textwidth}
\captionsetup{font=tiny}
\captionsetup{justification=centering}
 \centering
  \includegraphics[width=1\linewidth]{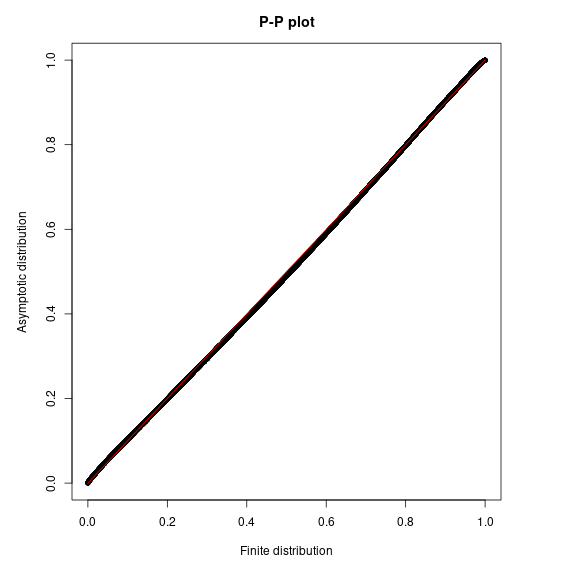}
    \caption*{Erd\H{o}s-R{\'e}nyi with infections, $(n,t)=(15,15)$.}
\end{subfigure}%
\begin{subfigure}{.275\textwidth}
\captionsetup{font=tiny}
\captionsetup{justification=centering}
 \centering
  \includegraphics[width=1\linewidth]{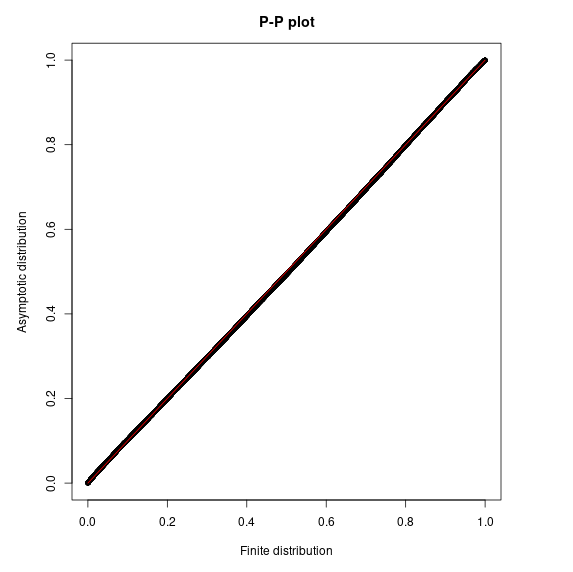}
    \caption*{Erd\H{o}s-R{\'e}nyi with infections, $(n,t)=(50,25)$.}
\end{subfigure}%
\begin{subfigure}{.275\textwidth}
\captionsetup{font=tiny}
\captionsetup{justification=centering}
 \centering
  \includegraphics[width=1\linewidth]{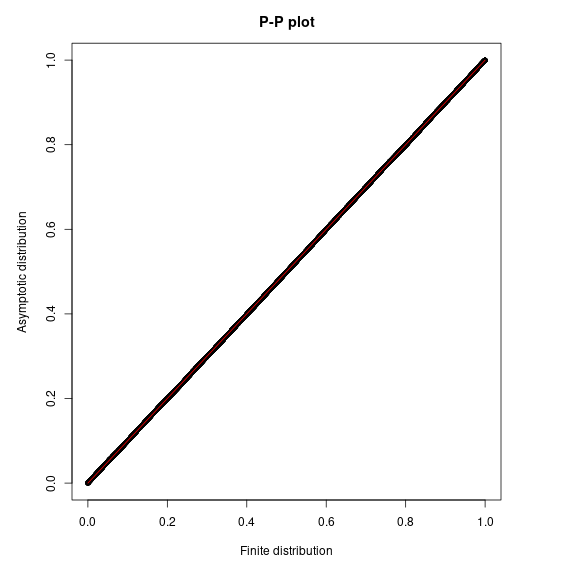}
    \caption*{Erd\H{o}s-R{\'e}nyi with infections, $(n,t)=(200,80)$.}
\end{subfigure}

\hspace{0.75cm}
\begin{subfigure}{.275\textwidth}
\captionsetup{font=tiny}
\captionsetup{justification=centering}
 \centering
  \includegraphics[width=1\linewidth]{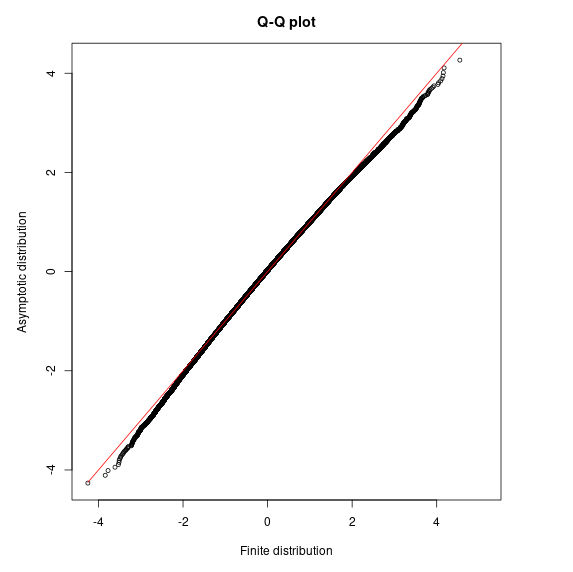}
    \caption*{Erd\H{o}s-R{\'e}nyi with infections, $(n,t)=(15,15)$.}
\end{subfigure}%
\begin{subfigure}{.275\textwidth}
\captionsetup{font=tiny}
\captionsetup{justification=centering}
 \centering
  \includegraphics[width=1\linewidth]{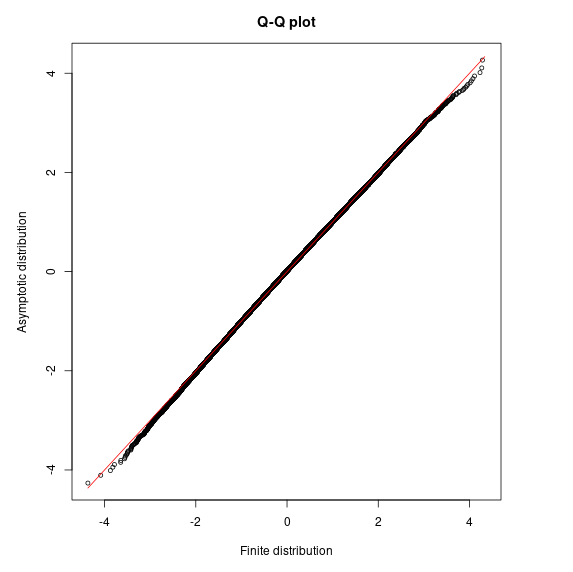}
    \caption*{Erd\H{o}s-R{\'e}nyi with infections, $(n,t)=(50,25)$.}
\end{subfigure}%
\begin{subfigure}{.275\textwidth}
\captionsetup{font=tiny}
\captionsetup{justification=centering}
 \centering
  \includegraphics[width=1\linewidth]{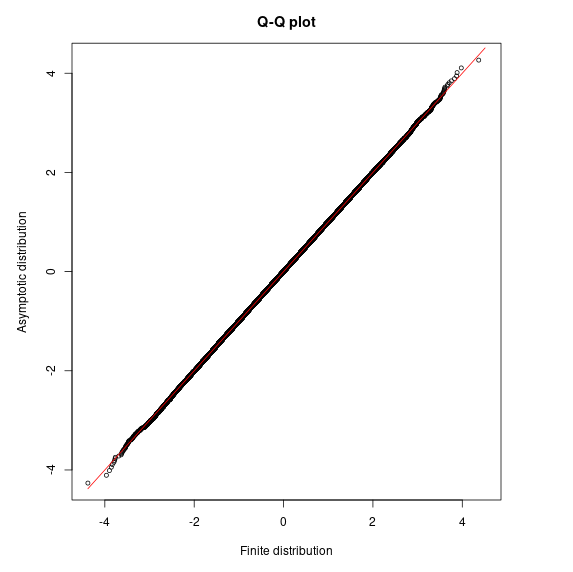}
    \caption*{Erd\H{o}s-R{\'e}nyi with infections, $(n,t)=(200,80)$.}
\end{subfigure}

\hspace{0.75cm}
\begin{subfigure}{.275\textwidth}
\captionsetup{font=tiny}
\captionsetup{justification=centering}
 \centering
  \includegraphics[width=1\linewidth]{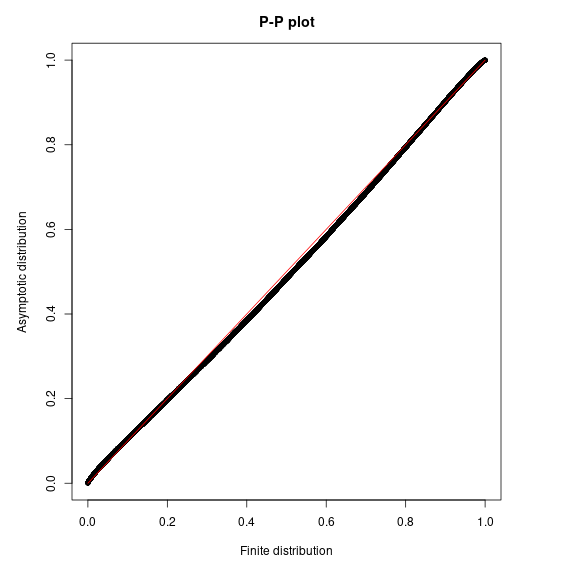}
    \caption*{Contagion case, $(n,t)=(15,15)$.}
\end{subfigure}%
\begin{subfigure}{.275\textwidth}
\captionsetup{font=tiny}
\captionsetup{justification=centering}
 \centering
  \includegraphics[width=1\linewidth]{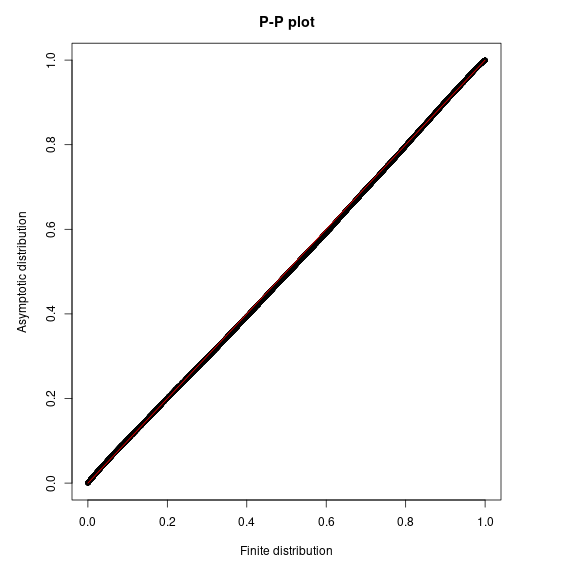}
    \caption*{Contagion case, $(n,t)=(50,25)$.}
\end{subfigure}%
\begin{subfigure}{.275\textwidth}
\captionsetup{font=tiny}
\captionsetup{justification=centering}
 \centering
  \includegraphics[width=1\linewidth]{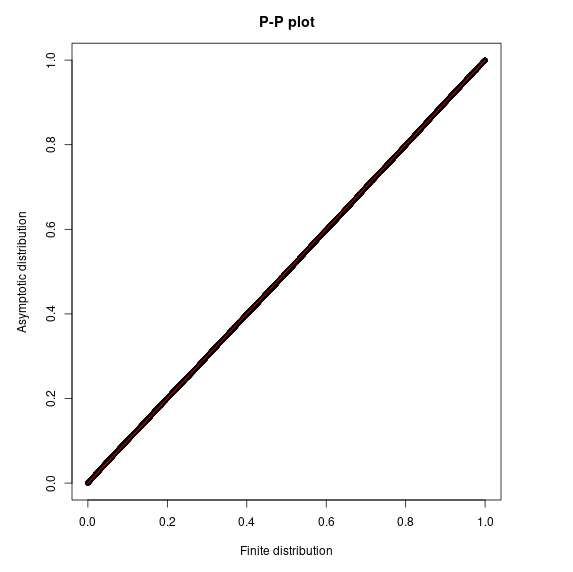}
    \caption*{Contagion case, $(n,t)=(200,80)$.}
\end{subfigure}

\hspace{0.75cm}
\begin{subfigure}{.275\textwidth}
\captionsetup{font=tiny}
\captionsetup{justification=centering}
 \centering
  \includegraphics[width=1\linewidth]{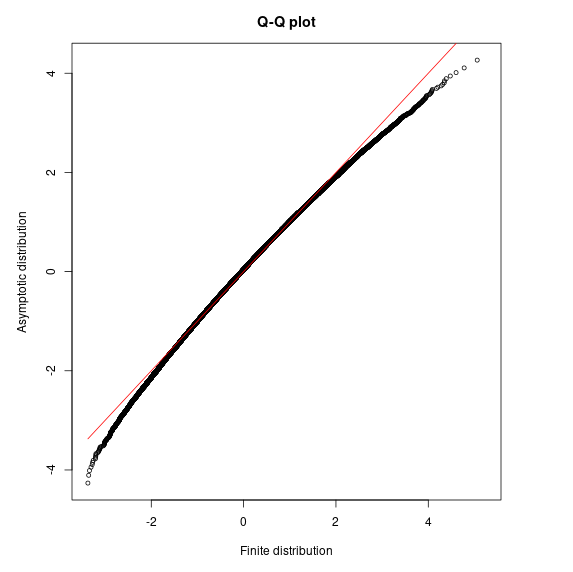}
    \caption*{Contagion case, $(n,t)=(15,15)$.}
\end{subfigure}%
\begin{subfigure}{.275\textwidth}
\captionsetup{font=tiny}
\captionsetup{justification=centering}
 \centering
  \includegraphics[width=1\linewidth]{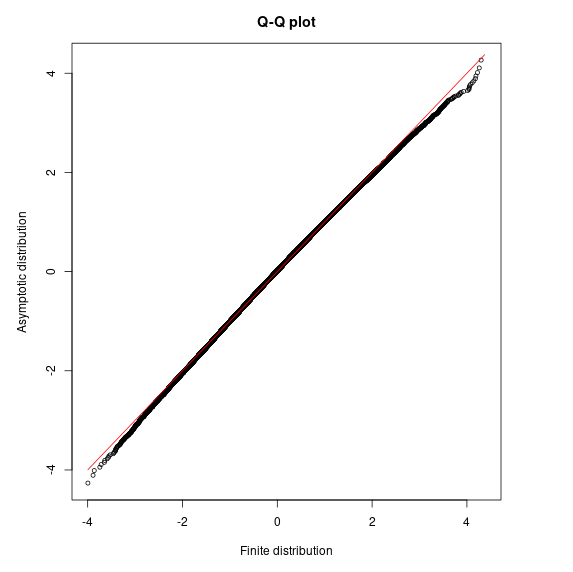}
    \caption*{Contagion case, $(n,t)=(50,25)$.}
\end{subfigure}%
\begin{subfigure}{.275\textwidth}
\captionsetup{font=tiny}
\captionsetup{justification=centering}
 \centering
  \includegraphics[width=1\linewidth]{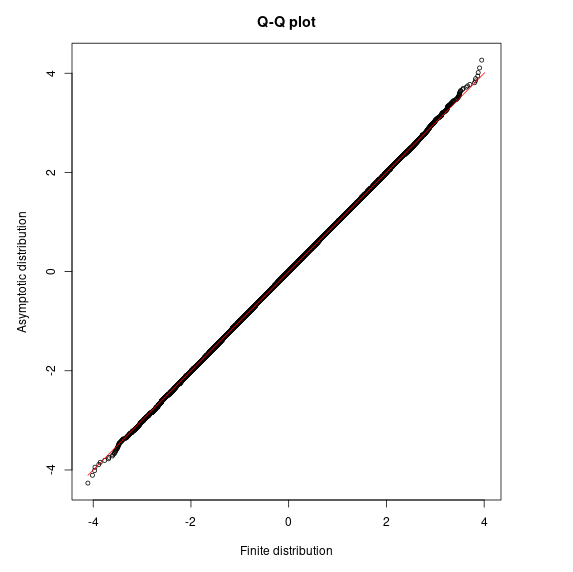}
    \caption*{Contagion case, $(n,t)=(200,80)$.}
\end{subfigure}
\caption{Plots related to Theorem \ref{thm3}.}
  \label{fig4}
\end{figure}

\begin{figure}[!t]
\hspace{0.75cm}
\begin{subfigure}{.275\textwidth}
\captionsetup{font=tiny}
\captionsetup{justification=centering}
 \centering
  \includegraphics[width=1\linewidth]{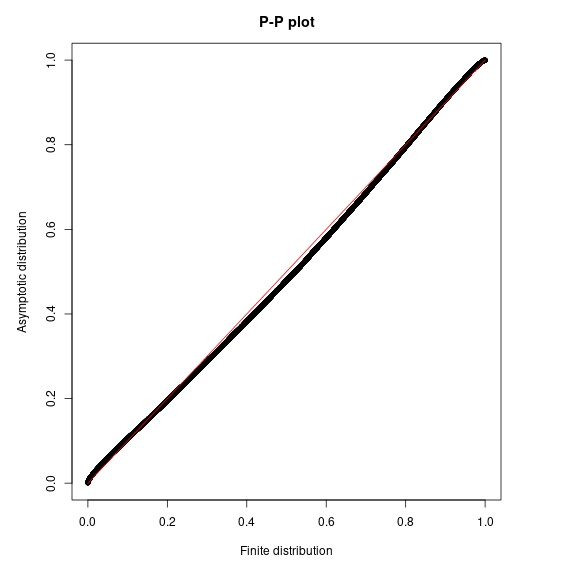}
    \caption*{Erd\H{o}s-R{\'e}nyi with infections, $t=15$.}
\end{subfigure}%
\begin{subfigure}{.275\textwidth}
\captionsetup{font=tiny}
\captionsetup{justification=centering}
 \centering
  \includegraphics[width=1\linewidth]{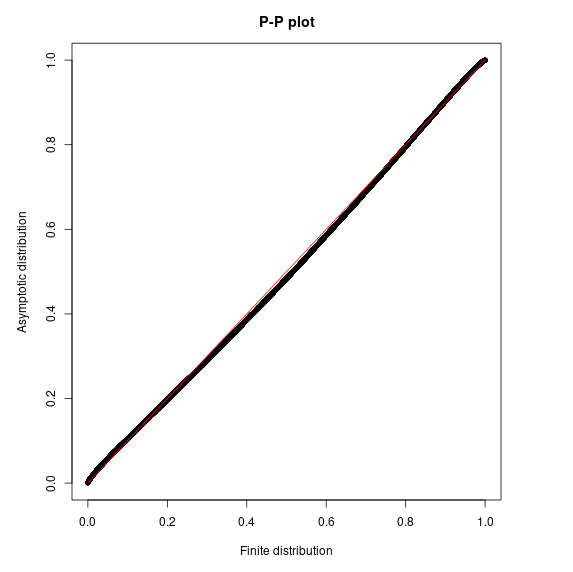}
    \caption*{Erd\H{o}s-R{\'e}nyi with infections, $t=25$.}
\end{subfigure}%
\begin{subfigure}{.275\textwidth}
\captionsetup{font=tiny}
\captionsetup{justification=centering}
 \centering
  \includegraphics[width=1\linewidth]{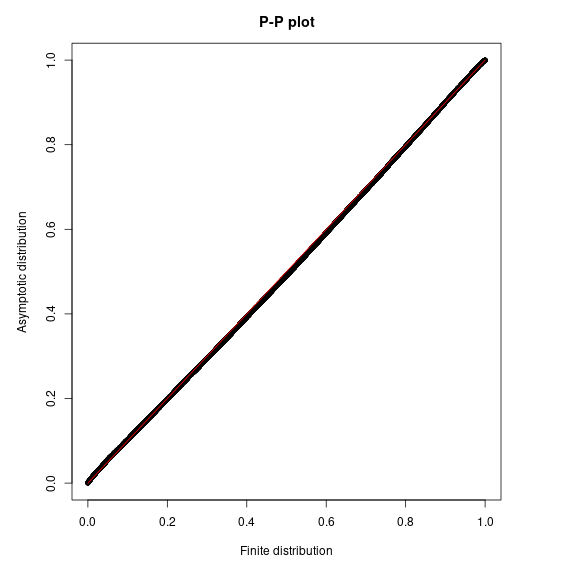}
    \caption*{Erd\H{o}s-R{\'e}nyi with infections, $t=80$.}
\end{subfigure}

\hspace{0.75cm}
\begin{subfigure}{.275\textwidth}
\captionsetup{font=tiny}
\captionsetup{justification=centering}
 \centering
  \includegraphics[width=1\linewidth]{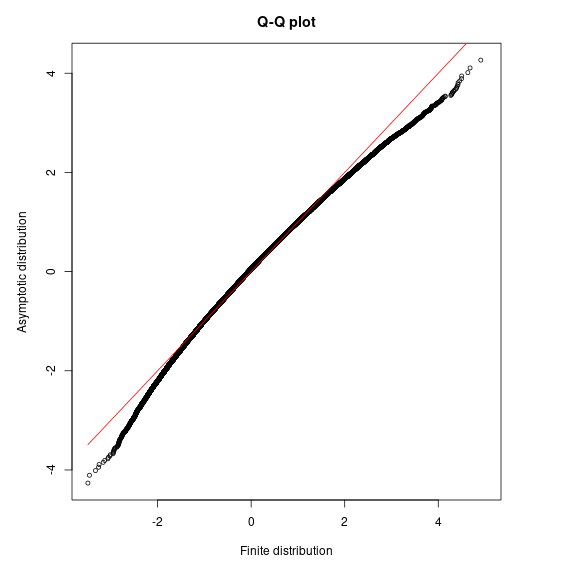}
    \caption*{Erd\H{o}s-R{\'e}nyi with infections, $t=15$.}
\end{subfigure}%
\begin{subfigure}{.275\textwidth}
\captionsetup{font=tiny}
\captionsetup{justification=centering}
 \centering
  \includegraphics[width=1\linewidth]{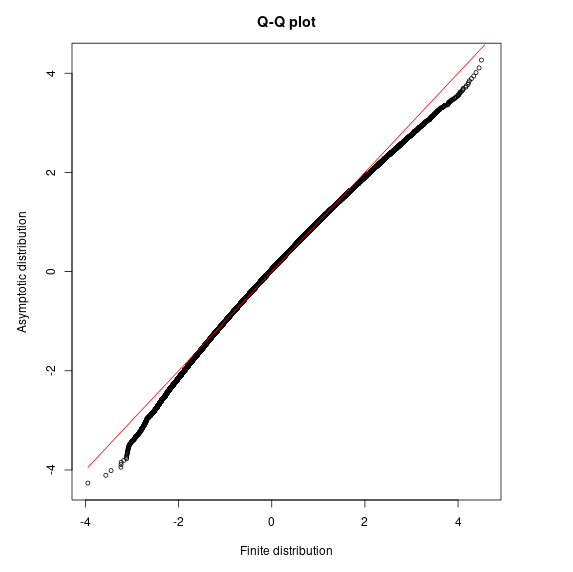}
    \caption*{Erd\H{o}s-R{\'e}nyi with infections, $t=25$.}
\end{subfigure}%
\begin{subfigure}{.275\textwidth}
\captionsetup{font=tiny}
\captionsetup{justification=centering}
 \centering
  \includegraphics[width=1\linewidth]{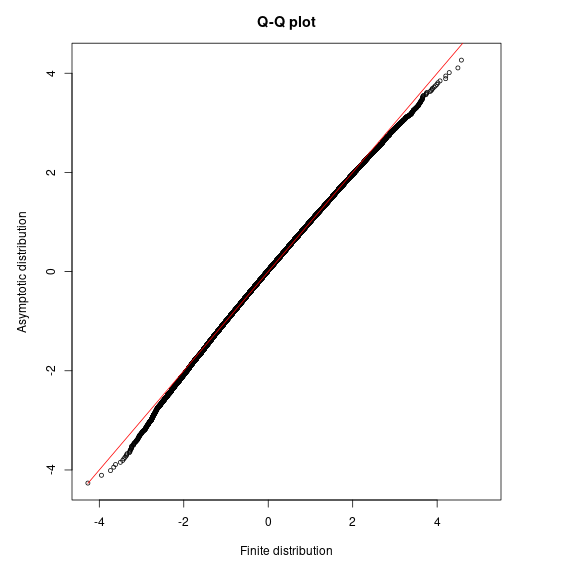}
    \caption*{Erd\H{o}s-R{\'e}nyi with infections, $t=80$.}
\end{subfigure}

\hspace{0.75cm}
\begin{subfigure}{.275\textwidth}
\captionsetup{font=tiny}
\captionsetup{justification=centering}
 \centering
  \includegraphics[width=1\linewidth]{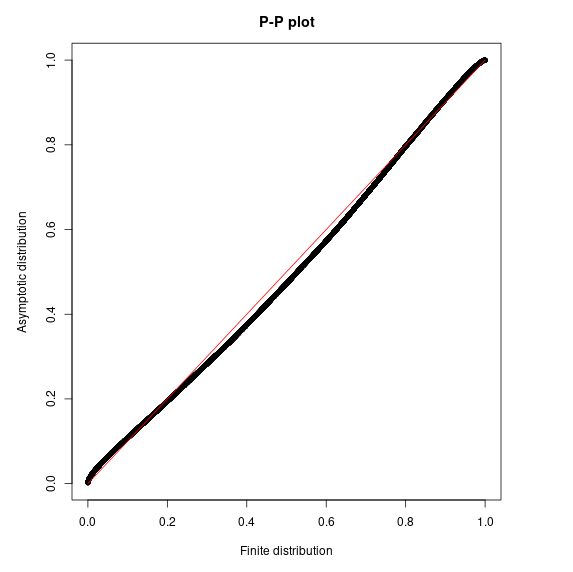}
    \caption*{Contagion case, $t=15$.}
\end{subfigure}%
\begin{subfigure}{.275\textwidth}
\captionsetup{font=tiny}
\captionsetup{justification=centering}
 \centering
  \includegraphics[width=1\linewidth]{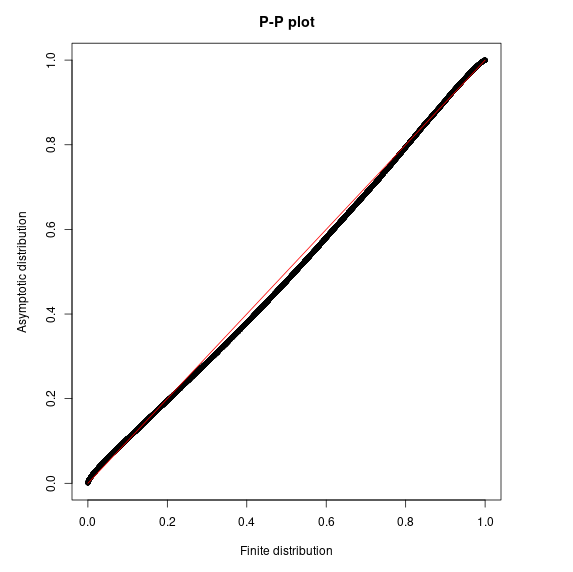}
    \caption*{Contagion case, $t=25$.}
\end{subfigure}%
\begin{subfigure}{.275\textwidth}
\captionsetup{font=tiny}
\captionsetup{justification=centering}
 \centering
  \includegraphics[width=1\linewidth]{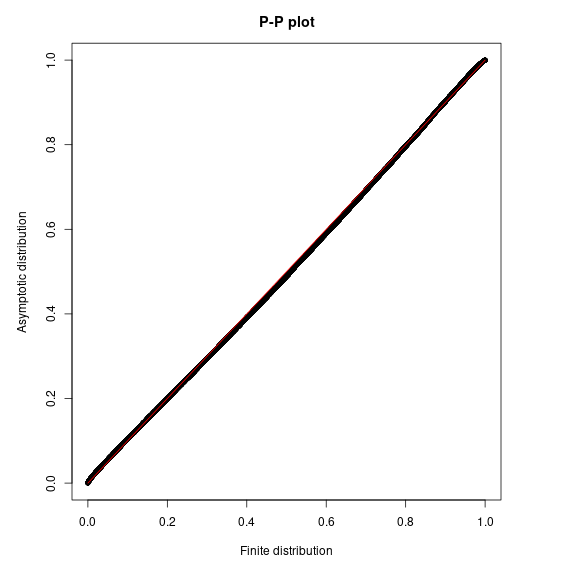}
    \caption*{Contagion case, $t=80$.}
\end{subfigure}

\hspace{0.75cm}
\begin{subfigure}{.275\textwidth}
\captionsetup{font=tiny}
\captionsetup{justification=centering}
 \centering
  \includegraphics[width=1\linewidth]{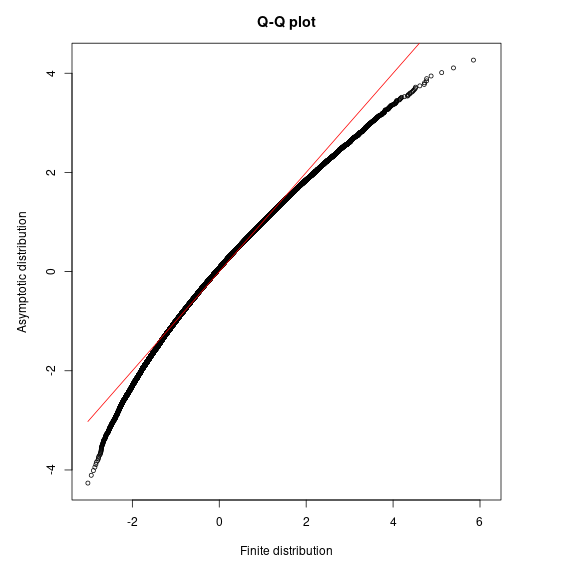}
    \caption*{Contagion case, $t=15$.}
\end{subfigure}%
\begin{subfigure}{.275\textwidth}
\captionsetup{font=tiny}
\captionsetup{justification=centering}
 \centering
  \includegraphics[width=1\linewidth]{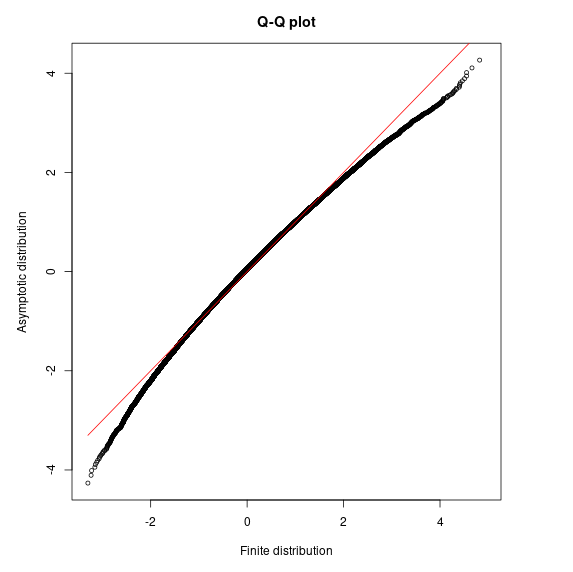}
    \caption*{Contagion case, $t=25$.}
\end{subfigure}%
\begin{subfigure}{.275\textwidth}
\captionsetup{font=tiny}
\captionsetup{justification=centering}
 \centering
  \includegraphics[width=1\linewidth]{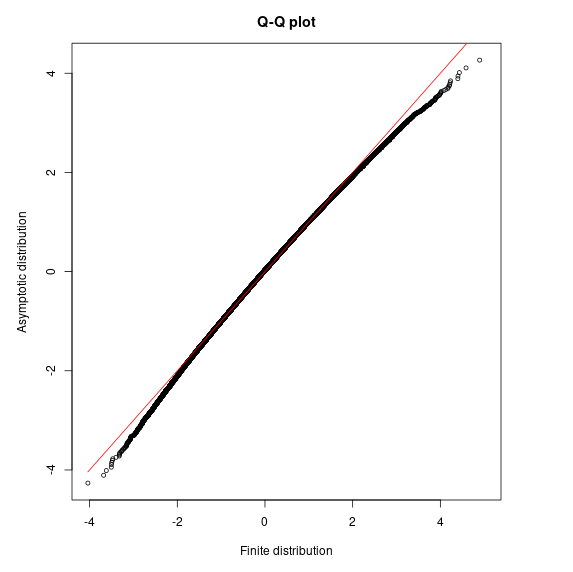}
    \caption*{Contagion case, $t=80$.}
\end{subfigure}
\caption{Plots related to Theorem \ref{thm5} a).}
  \label{fig5}
\end{figure}

\subsection{The impact of interaction}\label{comp}

We numerically compare the interaction-augmented model specified in Section~\ref{imp} with the classical collective risk model and with models that cover systemic and time-dependent interactions. To facilitate a direct comparison with a binomial claim count distribution, as described in Remark~\ref{rmk4}, we consider two loss events by conditioning on $N(t) = 2$ for a given point in time~$t$.\looseness=-1 

Using Monte Carlo simulations and asymptotic approximations from Section~\ref{theory}, we compare the model in Section~\ref{imp} to  alternative settings. The setup in Section~\ref{imp} is referred to as the \emph{dependent} case, where infections and transmissions $I$ follow either the Erd\H{o}s-R{\'e}nyi model with infections (Example~\ref{exa1}~d)) or the contagion model (Example~\ref{exa1}~e)). Positively dependent conditional loss expenditures $Z$ are structured according to equation~\eqref{lossexa}. Numerical results are presented in Table~\ref{tabconn} and Table~\ref{tabcont}. In both cases, we compare the model from Section~\ref{imp} to a corresponding \emph{independent} case, a \emph{standard} case,  a \emph{systemic} case and  a \emph{time-dependent} case, where the standard case remains identical across both tables - apart from minor differences due to Monte Carlo sampling.

The {independent} case uses the same model for infections and transmissions $I$ as the dependent case. Independence pertains to the conditional loss expenditures $Z$, which follow a modified version of equation~\eqref{lossexa}. Specifically, we set
\begin{align*}
Z_{i,j,k} = \widetilde Z_{i,j,k} + \varepsilon_{i,j,k}, \quad i,j,k \in \mathbb{N},
\end{align*}
with independent $\widetilde Z_{i,j,k} \overset{d}{=}\widetilde Z_{j,k}$, $i,j,k \in \mathbb{N}$. Thus, the contagious components $\widetilde Z_{j,k}$, $j,k \in \mathbb{N}$, are replaced by idiosyncratic loss expenditures with the same distribution. This can be interpreted as a collective model according to Theorem~\ref{thmrel}, where the claims count process is governed by the infections and transmissions $I$. The {standard} case further modifies the matrix $I$ in line with Example~\ref{exa1}~a), removing all transmissions and assuming independent infections. The corresponding claims count process follows a binomial distribution, as described in Remark~\ref{rmk4}. In the {systemic} case, we set
\begin{align*}
Z_{i,j,k} = \widetilde Z_{k} + \varepsilon_{i,j,k}, \quad i,j,k \in \mathbb{N},
\end{align*}
with independent $\widetilde Z_{k} \overset{d}{=}\widetilde Z_{j,k}$, $k \in \mathbb{N}$. Here, the contagious components $\widetilde Z_{j,k}$, $j,k \in \mathbb{N}$, are replaced by systemic loss expenditures with the same distribution. In the {time-dependent} case, we set
\begin{align*}
Z_{i,j,k} = \widetilde Z + \varepsilon_{i,j,k}, \quad i,j,k \in \mathbb{N},
\end{align*}
with  $\widetilde Z \overset{d}{=}\widetilde Z_{j,k}$. Now, the contagious components $\widetilde Z_{j,k}$, $j,k \in \mathbb{N}$, are replaced by a global variable with the same distribution.

\color{black}

The numerical results of the comparison, rounded to one decimal place, are presented in Table~\ref{tabconn} for the Erd\H{o}s-R{\'e}nyi model with infections and in Table~\ref{tabcont} for the contagion model. For the comparison, we focus on the upper quantiles of the distribution of the total loss of the insurance company, $S(n,t)$, which are particularly relevant in risk management. Specifically, we report the quantiles $q_{0.5}$, $q_{0.90}$, $q_{0.95}$, and $q_{0.995}$, corresponding to probability levels $0.5$, $0.90$, $0.95$, and $0.995$, respectively. The latter level is of particular importance in the context of the Solvency II directive. Additionally, we provide the associated quantile differences $q_{0.90} - q_{0.5}$, $q_{0.95} - q_{0.5}$, and $q_{0.995} - q_{0.5}$ to capture the dispersion in the tail of the distribution. Results are displayed for portfolio sizes $n = 25, 50, 75$, considering the time-dependent case,  the systemic case, the dependent case, the independent case and the standard case as defined above. {We use 100,000 Monte Carlo replications in each case.}

 \color{black}
Comparing the Erd\H{o}s-R{\'e}nyi model with infections in Table~\ref{tabconn} to the contagion model in Table~\ref{tabcont}, we observe that, in the time-dependent case,  the systemic case, the dependent case and the independent case, the quantiles are larger in the Erd\H{o}s-R{\'e}nyi model.  This is primarily due to the fact that contagion in the contagion model occurs only conditionally on an infection, whereas in the Erd\H{o}s-R{\'e}nyi model, transmissions directly trigger loss occurrences.

\color{black}
We now focus on observations within each of the two tables, which exhibit qualitatively similar patterns. 
{The differences from the median increase from the standard case through the independent, dependent, and systemic cases to the time-dependent case.} This is reasonable because the different cases model increasing correlations between the losses incurred by different policyholders, leading to greater risk exposure (i.e., higher upper quantiles relative to the median). The quantile distances also increase with higher quantile levels. 
\color{black}

The portfolio size naturally influences the number of infections and transmissions differently: infections increase linearly in $n$, while transmissions grow quadratically in $n$ in the Erd\H{o}s-R{\'e}nyi model. In the contagion model, these effects are intertwined, adding further complexity. Additionally, diversification effects are present, influencing loss dynamics as portfolio size increases. These combined factors contribute to the overall increase in total losses for larger portfolios.

All findings hold consistently for both the Erd\H{o}s-R{\'e}nyi model (Table~\ref{tabconn}) and the contagion model (Table~\ref{tabcont}). In the standard case, where neither contagious losses nor transmissions occur, quantiles are significantly smaller.

\begin{table}[H]
\color{black}
    \centering
    \singlespacing
    \renewcommand{\arraystretch}{1.65}
    \setlength{\tabcolsep}{3pt}
    \renewcommand\cellgape{\Gape[0pt]}
\begin{tabular}{ccccccccc}
$n$ & Case & $q_{0.5}$ & $q_{0.90}$ & $q_{0.95}$ & $q_{0.995}$ & $q_{0.90}-q_{0.5}$ & $q_{0.95}-q_{0.5}$ & $q_{0.995}-q_{0.5}$\\
\hline
\multirow{5}{*}{25} & {Time-dependent}  & 865.0 & 1930.8 & 2426.8 & 4068.4 & 1065.8 & 1561.7 & 3203.4\\
& {Systemic}   & 966.6 & 1712.4 & 2012.5 & 2955.2 & 745.8 & 1045.9 & 1988.6\\
& {Dependent}   & 1079.0 & 1236.2 & 1285.3 & 1417.4 & 157.2 & 206.3 & 338.4\\
& {Independent}  & 1084.8 & 1175.6 & 1201.0 & 1267.5 & 90.8 & 116.2 & 182.7\\
& {Standard}  & 30.7 & 44.4 & 48.8 & 60.6 & 13.7 & 18.0 & 29.9\\
\hline
\multirow{5}{*}{50} & {Time-dependent}  & 3465.0 & 7708.0 & 9679.1 & 16477.7 & 4243.0 & 6214.1 & 13012.7\\
& {Systemic}   & 3888.2 & 6859.5 & 8062.1 & 11878.3 & 2971.3 & 4173.9 & 7990.1\\
& {Dependent}   & 4355.5 & 4750.1 & 4869.9 & 5196.3 & 394.6 & 514.4 & 840.7\\
& {Independent}  & 4366.0 & 4545.7 & 4597.9 & 4733.2 & 179.6 & 231.9 & 367.2\\
& {Standard}  & 61.9 & 80.9 & 86.7 & 102.6 & 19.0 & 24.7 & 40.6\\
\hline
\multirow{5}{*}{75} & {Time-dependent}  & 7801.3 & 17443.8 & 21814.3 & 36375.9 & 9642.4 & 14013.0 & 28574.6\\
& {Systemic}   & 8729.5 & 15423.1 & 18112.3 & 26723.6 & 6693.5 & 9382.8 & 17994.1\\
& {Dependent}   & 9824.4 & 10520.8 & 10732.5 & 11286.4 & 696.4 & 908.1 & 1462.1\\
& {Independent}  & 9842.9 & 10113.3 & 10188.5 & 10386.2 & 270.4 & 345.6 & 543.3\\
& {Standard}  & 93.4 & 116.5 & 123.2 & 142.2 & 23.1 & 29.8 & 48.8\\
\end{tabular}
				\caption{\label{tabconn} Quantiles and related distances for the Erd\H{o}s-R{\'e}nyi model with infections.}

		\end{table}

\begin{table}[H]
\color{black}
    \centering
    \singlespacing
    \renewcommand{\arraystretch}{2}
    \setlength{\tabcolsep}{3pt}
    \renewcommand\cellgape{\Gape[0pt]}
			
						\begin{tabular}{ccccccccc}
$n$ & Case & $q_{0.5}$ & $q_{0.90}$ & $q_{0.95}$ & $q_{0.995}$ & $q_{0.90}-q_{0.5}$ & $q_{0.95}-q_{0.5}$ & $q_{0.995}-q_{0.5}$\\
\hline
\multirow{5}{*}{25} & {Time-dependent}  & 233.6 & 539.3 & 684.2 & 1200.0 & 305.7 & 450.6 & 966.4\\
& {Systemic}   & 257.8 & 489.7 & 591.9 & 935.2 & 231.9 & 334.1 & 677.4\\
& {Dependent}   & 287.6 & 415.6 & 457.0 & 566.6 & 128.0 & 169.3 & 279.0\\
& {Independent}  & 291.7 & 396.7 & 427.4 & 509.2 & 105.0 & 135.7 & 217.5\\
& {Standard}  & 30.6 & 44.4 & 48.7 & 60.3 & 13.7 & 18.1 & 29.7\\
\hline
\multirow{5}{*}{50} & {Time-dependent}  & 901.5 & 2044.7 & 2575.4 & 4473.6 & 1143.2 & 1673.9 & 3572.1\\
& {Systemic}   & 1000.2 & 1829.9 & 2176.9 & 3317.7 & 829.8 & 1176.7 & 2317.5\\
& {Dependent}   & 1126.1 & 1460.6 & 1565.3 & 1845.4 & 334.4 & 439.1 & 719.3\\
& {Independent}  & 1132.6 & 1403.0 & 1481.3 & 1693.3 & 270.4 & 348.6 & 560.7\\
& {Standard}  & 62.0 & 81.0 & 86.9 & 102.4 & 19.1 & 25.0 & 40.4\\
\hline
\multirow{5}{*}{75} & {Time-dependent}  & 1996.5 & 4499.7 & 5674.4 & 9796.0 & 2503.1 & 3677.9 & 7799.4\\
& {Systemic}   & 2236.6 & 4042.6 & 4800.6 & 7321.2 & 1805.9 & 2564.0 & 5084.5\\
& {Dependent}   & 2510.7 & 3111.1 & 3288.0 & 3782.1 & 600.5 & 777.3 & 1271.4\\
& {Independent}  & 2524.4 & 3003.3 & 3142.0 & 3514.2 & 478.9 & 617.6 & 989.8\\
& {Standard}  & 93.3 & 116.5 & 123.4 & 142.8 & 23.1 & 30.1 & 49.5\\
\end{tabular}
				\caption{\label{tabcont} Quantiles and related distances for the contagion model.}

		\end{table}

\color{black}

\section{Conclusion and outlook}\label{con}

In an insurance context, we extended collective models to incorporate interactions in loss occurrences and expenditures. The key feature of our approach is the joint exchangeability of arrays (which is weaker than de Finetti's exchangeability), which allows us to move beyond the assumption of independence. Within this framework, we established limit theorems as the number of insurance contracts, the time horizon, or both tend to infinity. Furthermore, we rigorously examined the relationship to the classical collective model. To support our findings, we implemented the model in the statistical software R and conducted simulation studies to assess the accuracy of the approximations and evaluate the impact of network interactions on total loss distributions.

From a more practical perspective, building on \citet{zs22}, the model could be extended to incorporate additional covariates, thereby enhancing its applicability. Such extensions would enable its use in insurance pricing and risk management, particularly in the context of cyber and power networks. On the theoretical side, further research directions include  assessing the applicability of finite sample corrections to improve approximation accuracy. Additionally, statistical methodologies such as parameter estimation and goodness-of-fit assessment (i.e., model calibration and validation) based on loss data should be explored both theoretically and in practical applications.

Beyond these extensions, several fundamental challenges remain. A key question is the development of interaction mechanisms beyond joint exchangeability. Another direction is extending the framework from univariate to multivariate loss processes that would enable the joint modeling of different risk types, such as operational, cyber, and financial losses, along with their interdependencies. Finally, systemic risk modeling remains a major challenge, as losses may propagate not only within a single insurance portfolio but also across multiple insurers and financial institutions, requiring systemic contagion models that account for inter-firm dependencies, capital constraints, and regulatory spillover effects.

\begin{appendix}

\section{Proofs}\label{proofs}

\color{black}

\begin{proof}[Proof of Theorem \ref{thm2}.]
At first, we show the stated  convergence in distribution conditioned on  $\xi, \xi_{k},\zeta_{k},$ $k\in\N$, almost surely. Then, we deduce the stated {unconditional convergence in distribution}. For the convergence in distribution conditioned on  $\xi, \xi_{k},\zeta_{k},$ $k\in\N$, almost surely we consider constants  $x,x_k,z_k \in[0,1]$, $k\in\N$, with $\xi=x, \xi_{k}=x_k,\zeta_{k}=z_k,$ $k\in\N$, such that
$$G_{i,j,k} = h(x,x_k,\xi_{i,k}, \xi_{j,k}, \xi_{i,j,k}), \quad (i,j,k) \in \N \times \N \times \N,$$
where  $\xi_{j,i,k} = \xi_{i,j,k}$ for $i\leq j${. The loss event times satisfy}
$$T_1=\ell_1(z_1,x),~T_k=\ell_k(z_k,x,T_1,\dots,T_{k-1}),~k\ge 2.$$
Then, the random variable $N(t)$ is  constant, $N(t)=m\in\N_0$, say, and  the arrays of random variables $G_k$, $k\in \N$, {are independent of one another, and each is dissociated}. We focus on the case $N(t)=m\ge 1$, noticing that in the case $N(t)=m=0$ we obtain  $\frac{S(n,t)-\mu(n,t)}{n^{ 3/ 2 }}=0\sim N(0,0)=N(0,\sigma^2(t))$ for all $n\in\N$ and the  convergence in distribution follows immediately.  For  $1 \leq i , j \leq n$, $i \neq j$,  define
\begin{align*}
\widetilde X_{i,j}= X_{i,j} (t)+ X_{j,i} (t)- \E( X_{i,j} (t)+ X_{j,i}(t)),~\widetilde X_{i,i}= X_{i,i}(t)- \E(X_{i,i}(t)).
\end{align*}
Remember  that $X_{i,j}(t)=\sum_{k=1}^{N(t)}G_{i,j,k}=\sum_{k=1}^{m}G_{i,j,k}$ and note that we do not use the point in time $t$ consistently in our notation  because the point in time is constant here.  The array $\widetilde X = (\widetilde X_{i,j})_{(i,j) \in \mathbb{N} \times \mathbb{N}}$ is centered, symmetric, jointly exchangeable, and dissociated. Moreover, we can decompose
\begin{align*}
 S(n,t)-\mu(n,t)= U(n) + R(n),
\end{align*}
where
\begin{align*}
U(n) = \sum_{1 \leq i < j \leq n} \widetilde X_{i,j}, ~R(n) = \sum_{i=1}^{n} \widetilde X_{i,i}.
\end{align*}
By Theorem A in \citet{sil}, the statistic $U(n)$ satisfies
\begin{align*}
\frac{U(n)}{n^{ 3/ 2 }}\overset{d}{\longrightarrow}N(0,\sigma^2(t))~\text{as}~n\to\infty.
\end{align*}
Furthermore, we have by the  strong law of large numbers
\begin{align*}
\frac{R(n)}{n}\overset{a.s.}{\longrightarrow}0~\text{as}~n\to\infty.
\end{align*}
{Since} 
\begin{align*}
\frac{S(n,t)-\mu(n,t)}{n^{ 3/ 2 }}= \frac{U(n)}{n^{ 3/ 2 }}+\frac{n}{n^{ 3/ 2 }}\frac{R(n)}{n},
\end{align*}
Slutsky's theorem implies 
\begin{align*}
\frac{S(n,t)-\mu(n,t)}{n^{ 3/ 2 }}\overset{d}{\longrightarrow}N(0,\sigma^2(t))~\text{as}~n\to\infty.
\end{align*}
Let us fix some arbitrary $s\in\R$ in what follows. We deduce that 
\begin{align*}
h_n(\xi,\xi_k,\zeta_k,k\in\N)=P\Bigg(\frac{S(n,t)-\mu(n,t)}{n^{ 3/ 2 }}\le s\Bigg|\xi, \xi_{k},\zeta_{k},k\in\N \Bigg)\longrightarrow\Phi_{0,\sigma^2(t)}(s)~\text{as}~n\to\infty
\end{align*}
almost surely. Now, we show the stated {unconditional convergence in distribution}. Observing
\begin{align*}
P\Bigg(\frac{S(n,t)-\mu(n,t)}{n^{ 3/ 2 }}\le s\Bigg)=\E\Big(h_n\big(\xi, \xi_{k},\zeta_{k},k\in\N \big)\Big),
\end{align*} 
it follows from the dominated convergence theorem, since $h_n\ge 0$ and $h_n$ is bounded by 1, that 
\begin{align*}
\E(h_n\big(\xi, \xi_{k},\zeta_{k},k\in\N \big))\longrightarrow \E(\Phi_{0,\sigma^2(t)}(s))~\text{as}~n\to\infty.
\end{align*}
\end{proof}

\begin{proof}[Proof of Theorem \ref{thm3}.]

Firstly, we prove the stated convergence in distribution conditioned on  $\xi, \zeta_{k},$ $k\in\N$, almost surely. Subsequently, we {deduce the unconditional convergence in distribution}. For the convergence in distribution conditioned on  $\xi,\zeta_{k},$ $k\in\N$, almost surely we consider  constants  $x,z_k \in[0,1]$, $k\in\N$, with $\xi=x,\zeta_{k}=z_k,$ $k\in\N$, such that for the  arrays of random variables $G_k$, $k\in \N$,  
$$G_{i,j,k} = h(x,\xi_k,\xi_{i,k}, \xi_{j,k}, \xi_{i,j,k}), \quad (i,j,k) \in \N \times \N \times \N,$$
where  $\xi_{j,i,k} = \xi_{i,j,k}$ for $i\leq j${. The loss event times satisfy} 
$$T_1=\ell_1(z_1,x),~T_k=\ell_k(z_k,x,T_1,\dots,T_{k-1}),~k\ge 2,$$
with $\lim_{t\to\infty}N(t)=\infty$. Then,  the counting process $N(t)$ is  constant, $N(t)=m(t)$, say, and  the arrays of random variables $G_k$, $k\in \N$, are independent and identically distributed. Because we treat the simultaneous convergence $\min(n,t)\to\infty$, we consider the counting process as an arbitrary sequence $m(n)\in\N$, $n\in\N$, with $\lim_{n\to\infty}m(n)=\infty$ and reduce the simultaneous convergence $\min(n,t)\to\infty$ to the convergence $n\to\infty$. Remember that $S(n,t)=\sum_{k=1}^{N(t)}L_{k}(n)=\sum_{k=1}^{m(n)}L_{k}(n)$, where $L_{k}(n)=\sum_{i,j=1}^n G_{i,j,k}$, $k\in\N$, {are independent and identically distributed for each fixed $n$}. It is 
\begin{align*}
\frac{S(n,t)-\mu(n,t)}{\sigma(n,t)}&=\frac{\sum_{k=1}^{m(n)} (L_k(n)-\E(L_k(n)|\xi_{k}))}{\sqrt{m(n)\Var(L_1(n)-\E(L_1(n)|\xi_{1})) }}
=\frac{S(n)}{s(n)},
\end{align*}
where $s^2(n)=m(n)$ and
\begin{align*}
S(n)=\sum_{k=1}^{m(n)}H_k(n)
\end{align*}
{with independent and identically distributed terms in each row}
\begin{align*}
H_k(n)=\frac{L_k(n)-\E(L_k(n)|\xi_{k})}{\sqrt{\Var(L_1(n)-\E(L_1(n)|\xi_{1}))}},~k=1,\dots,m(n),~n\in\N,
\end{align*}
that are normalized and also centered due to
\begin{align*}
\E(L_1(n))-\E(\E(L_1(n)|\xi_{1}))=\E(L_1(n))-\E(L_1(n))=0.
\end{align*}
We  will apply the central limit theorem of Lindeberg-Feller. For this purpose, we verify Lindeberg's condition. We have  for all $u\in(0,\infty)$ 
\begin{align*}
\frac{1}{s(n)^2}\sum_{k=1}^{m(n)}\E\big(H_k(n)^2I(|H_k(n)|>s(n) u)\big)=\E\big(H_1(n)^2I(|H_1(n)|>\sqrt{m(n)}u)\big).
\end{align*}
We can argue analogously to the proof of Theorem  \ref{thm2} and apply again Theorem A in \citet{sil} to obtain the convergence in distribution of $(L_1(n)-\E(L_1(n)|\xi_{1}))/n^{3/2}$ as $n\to\infty$  to a centered random variable with variance $\E(\sigma^2)$. Moreover, we have that $\Var(L_1(n)-\E(L_1(n)|\xi_{1}))/n^3$ converges to $\E(\sigma^2)$ as $n\to\infty$. In all, we obtain the convergence in distribution of $H_1(n)$  to a centered random variable $H$ with $\Var(H)=1$ as $n\to\infty$. By the Skorokhod representation theorem, we may replace $(H_1(n))_n$ and $H$ by new random variables (with notation unchanged) on a common probability space such that  $H_1(n)\overset{a.s.}{\to}H$ as $n\to\infty$. From 
\begin{align*}
H_1(n)^2I(|H_1(n)|>\sqrt{m(n)} u)\overset{a.s.}{\longrightarrow} 0~\text{as}~n\to\infty
\end{align*}
and $H_1(n)^2I(|H_1(n)|>\sqrt{m(n)} u)\le H_1(n)^2$ as well as $E(H_1(n)^2)=E(H^2)=1$ it follows from Pratt's version of the dominated convergence theorem that
\begin{align*}
E\big(H_1(n)^2I(|H_1(n)|>\sqrt{m(n)} u)\big)\longrightarrow 0~\text{as}~n\to\infty,
\end{align*}
i.e., Lindeberg's condition. The central limit theorem of Lindeberg-Feller implies
\begin{align*}
\frac{S(n)}{s(n)} \overset{d}{\longrightarrow}N(0,1)~\text{as}~n\to\infty.
\end{align*}
Let $s\in\R$ be arbitrary. We deduce that 
\begin{align*}
P\Bigg(\frac{S(n,t)-\mu(n,t)}{\sigma(n,t)}\le s\Bigg|\xi, \zeta_{k},k\in\N \Bigg)\longrightarrow \Phi_{0,1}(s)~\text{as}~\min(n,t)\to\infty
\end{align*}
almost surely, that is the stated convergence in distribution conditioned on  $\xi, \zeta_{k},$ $k\in\N$, almost surely.  The stated {unconditional convergence in distribution} follows analogously to the proof of Theorem \ref{thm2}. 
\end{proof}

\begin{proof} [Proof of Theorem \ref{thm4}.]

We prove again the stated convergence in distribution conditioned on  $\xi, \zeta_{k},$ $k\in\N$, almost surely at first and {deduce the unconditional convergence in distribution} later. For the convergence in distribution conditioned on  $\xi,\zeta_{k},$ $k\in\N$, almost surely we consider again constants  $x,z_k \in[0,1]$, $k\in\N$, with $\xi=x,\zeta_{k}=z_k,$ $k\in\N$, such that for the  arrays of random variables $G_k$, $k\in \N$,  
$$G_{i,j,k} = h(x,\xi_k,\xi_{i,k}, \xi_{j,k}, \xi_{i,j,k}), \quad (i,j,k) \in \N \times \N \times \N,$$
where  $\xi_{j,i,k} = \xi_{i,j,k}$ for $i\leq j${. The loss event times satisfy} 
$$T_1=\ell_1(z_1,x),~T_k=\ell_k(z_k,x,T_1,\dots,T_{k-1}),~k\ge 2,$$
with $\lim_{t\to\infty}N(t)=\infty$. Then,  the counting process $N(t)$ is again constant, $N(t)=m(t)$, say, and  the arrays of random variables $G_k$, $k\in \N$, are again independent and identically distributed.  Remember that $S(n,t)=\sum_{k=1}^{N(t)}L_{k}(n)=\sum_{k=1}^{m(t)}L_{k}(n)$, where $L_{k}(n)=\sum_{i,j=1}^n G_{i,j,k}$, $k\in\N$, are now independent and identically distributed random variables. We have

\begin{align*}
\frac{S(n,t)-\mu(n,t)}{\sigma(n,t)}&=\frac{1}{\sqrt{m(t)}}\sum_{k=1}^{m(t)} \frac{L_k(n)-\E(L_k(n))}{\sqrt{\Var(L_1(n)-\E(L_1(n))) }},
\end{align*}
where the standardized random variables
\begin{align*}
 \frac{L_k(n)-\E(L_k(n))}{\sqrt{\Var(L_1(n)-\E(L_1(n))) }},~k\in\N,
\end{align*}
are independent and identically distributed. {The Lindeberg--L\'evy central limit theorem gives}
\begin{align*}
\frac{S(n,t)-\mu(n,t)}{\sigma(n,t)}\overset{d}{\longrightarrow}N(0,1)~\text{as}~t\to\infty.
\end{align*}
Let us fix some arbitrary $s\in\R$ in what follows. We deduce that 
\begin{align*}
P\Bigg(\frac{S(n,t)-\mu(n,t)}{\sigma(n,t)}\le s\Bigg|\xi=x, \zeta_{k}=z_k,k\in\N \Bigg)\longrightarrow \Phi(s) ~\text{as}~t\to\infty
\end{align*}
for almost all $x,z_k \in[0,1]$, $k\in\N$, in other words
\begin{align*}
P\Bigg(\frac{S(n,t)-\mu(n,t)}{\sigma(n,t)}\le s\Bigg|\xi,\zeta_{k},k\in\N \Bigg)\longrightarrow \Phi(s) ~\text{as}~t\to\infty
\end{align*}
almost surely, that is the stated convergence in distribution conditioned on  $\xi, \zeta_{k},$ $k\in\N$, almost surely. The stated {unconditional convergence in distribution} follows analogously to the proof of Theorem \ref{thm3}. 
\end{proof}

\begin{proof} [Proof of Theorem \ref{thm40}.]
a) We show the stated bound conditioned on  $\zeta_{k},$ $k\in\N$, almost surely at first. Then, we deduce the {stated unconditional bound}. For the bound conditioned on  $\zeta_{k},$ $k\in\N$, almost surely we consider  constants  $z_k \in[0,1]$, $k\in\N$, with $\zeta_{k}=z_k,$ $k\in\N$, such that the points in time of the loss events in the counting process $N(t)$ are given by 
$$T_1=\ell_1(z_1),~T_k=\ell_k(z_k,T_1,\dots,T_{k-1}),~k\ge 2.$$ Note  that {the random variable} $N(t)$ is  now constant, $N(t)=m\in\N_0$, say, and  that the arrays of random variables $G_k$, $k\in \N$, {are independent of one another, and each is dissociated}.  In the case $N(t)=m=0$ we obtain  $S(n,t)=0\sim N(0,0)=N(m\mu_L(n),m\sigma_L^2(n))$ for all $n\in\N${, so the conditional bound holds with both sides equal to zero.} For that reason, let $N(t)=m\ge 1$.   We will use the notations in the proof of Theorem \ref{thm2} in what follows.  Note again that we do not use the point in time $t$ consistently in our notation. {In particular, we have}
\begin{align*}
\frac{S(n,t)-m\mu_L(n)}{\sqrt{m}\sigma_L(n)}=\frac{S(n,t)-\mu(n,t) }{\sigma(n,t)}=\frac{U(n)}{\sigma(n,t)}+\frac{R(n)}{\sigma(n,t)}.
\end{align*}
For all  $\varepsilon >0$ and all $s\in\R$, it holds that
\begin{align*}
P\bigg(\frac{U(n)}{ \sigma(n,t)}+\frac{R(n)}{ \sigma(n,t)}\le s\bigg)=&P\bigg(\frac{U(n)}{ \sigma(n,t)}+\frac{R(n)}{ \sigma(n,t)}\le s,\bigg|\frac{R(n)}{ \sigma(n,t)}\bigg|\le \varepsilon\bigg)\\
&+P\bigg(\frac{U(n)}{ \sigma(n,t)}+\frac{R(n)}{ \sigma(n,t)}\le s,\bigg|\frac{R(n)}{ \sigma(n,t)}\bigg|> \varepsilon\bigg)\\
\le & P\bigg(\frac{U(n)}{ \sigma(n,t)}\le s+\varepsilon\bigg)+P\bigg(\bigg|\frac{R(n)}{ \sigma(n,t)}\bigg|> \varepsilon\bigg).
\end{align*}
on the one hand and
\begin{align*}
P\bigg(\frac{U(n)}{ \sigma(n,t)}\le s-\varepsilon\bigg)=&P\bigg(\frac{U(n)}{ \sigma(n,t)}\le s-\varepsilon,\bigg|\frac{R(n)}{ \sigma(n,t)}\bigg|\le \varepsilon\bigg)\\
&+P\bigg(\frac{U(n)}{ \sigma(n,t)}\le s-\varepsilon,\bigg|\frac{R(n)}{ \sigma(n,t)}\bigg|> \varepsilon\bigg)\\
\le& P\bigg(\frac{U(n)}{ \sigma(n,t)}\le s-\frac{R(n)}{ \sigma(n,t)}\bigg)+P\bigg(\bigg|\frac{R(n)}{ \sigma(n,t)}\bigg|> \varepsilon\bigg)
\end{align*}
on the other hand, in all
\begin{align*}
P\bigg(\frac{U(n)}{ \sigma(n,t)}\le s-\varepsilon\bigg)-P\bigg(\bigg|\frac{R(n)}{ \sigma(n,t)}\bigg|> \varepsilon\bigg)&\le P\bigg(\frac{U(n)}{ \sigma(n,t)}+\frac{R(n)}{ \sigma(n,t)}\le s\bigg)\\
&\le  P\bigg(\frac{U(n)}{ \sigma(n,t)}\le s+\varepsilon\bigg)+P\bigg(\bigg|\frac{R(n)}{ \sigma(n,t)}\bigg|> \varepsilon\bigg).
\end{align*}
With $\sigma_U^2(n)=\Var(U(n))$ and  $r^2(n)=\frac{\sigma^2(n,t)}{\sigma_U^2(n)}$, it follows that
\begin{align*}
&P\bigg(\frac{U(n)}{\sigma_U(n)}\le r(n)(s-\varepsilon)\bigg)-\Phi_{0,1}(r(n)(s-\varepsilon))+\Phi_{0,1}(r(n)(s-\varepsilon))-\Phi_{0,1}(s)-P\bigg(\bigg|\frac{R(n)}{\sigma(n,t)}\bigg|> \varepsilon\bigg)\\
\le& P\bigg(\frac{U(n)}{\sigma(n,t)}+\frac{R(n)}{\sigma(n,t)}\le s\bigg)-\Phi_{0,1}(s)\\
\le  &P\bigg(\frac{U(n)}{\sigma_U(n)}\le r(n)(s+\varepsilon)\bigg)-\Phi_{0,1}(r(n)(s+\varepsilon))+\Phi_{0,1}(r(n)(s+\varepsilon))-\Phi_{0,1}(s)+P\bigg(\bigg|\frac{R(n)}{\sigma(n,t)}\bigg|> \varepsilon\bigg)
\end{align*}
and so
\begin{align*}
&\sup_{s\in\R}\bigg| P\bigg(\frac{S(n,t)-\mu(n,t)}{\sigma(n,t)}\le s\bigg)-\Phi_{0,1}(s) \bigg|\\
\le& \sup_{s\in\R}\bigg|  P\bigg(\frac{U(n)}{\sigma_U(n)}\le s\bigg)-\Phi_{0,1}(s) \bigg|\\
&+\max \bigg(  \sup_{s\in\R}\bigg|  \Phi_{0,1}(r(n)(s-\varepsilon))-\Phi_{0,1}(s)    \bigg| ,   \sup_{s\in\R}\bigg|  \Phi_{0,1}(r(n)(s+\varepsilon))-\Phi_{0,1}(s)    \bigg|   \bigg)\\
&+P\bigg(\bigg|\frac{R(n)}{\sigma(n,t)}\bigg|> \varepsilon\bigg).
\end{align*}
Let us treat all three summands on the right hand side of the latter inequality. {For each event count $m\ge1$, write $K_m=\widetilde X_{1,2}/\sqrt m$
and let $K_m^{(1)}=\E(K_m\mid\xi_{1,1},\ldots,\xi_{1,m})$ be its
projection onto the first vertex. Both $L^2$ norms are independent of $m$,
and the fourth-moment formula for independent sums gives
$\|K_m^{(1)}\|_4\le\|K_1\|_4$, so the proof of Theorem~2.1
in \citet{zha} yields}
\begin{align*}
\sup_{s\in\R}\bigg|  P\bigg(\frac{U(n)}{\sigma_U(n)}\le s\bigg)-\Phi_{0,1}(s) \bigg|\le \frac{1}{\sqrt n}\eta^2,
\end{align*}
where $\eta^2=\frac{16(\E((G_{1,2}+G_{2,1}-\E(G_{1,2}+G_{2,1}))^4){\color{black})}^{1/2}}{\Cov( G_{1,2}+G_{2,1}, G_{1,3}+G_{3,1})}$. {For the second term, the mean value theorem for integrals and elementary estimates give}
\begin{align*}
& \sup_{s\in\R}|  \Phi_{0,1}(r(n)(s-\varepsilon))-\Phi_{0,1}(s)    |\\
\le&  \sup_{s\in\R}|  \Phi_{0,1}(r(n)(s-\varepsilon))-\Phi_{0,1}((s-\varepsilon))    |+ \sup_{s\in\R}|  \Phi_{0,1}(s-\varepsilon)-\Phi_{0,1}(s)  |\\
\le&o(\frac{1}{\sqrt{n}})+\frac{\varepsilon}{\sqrt{2\pi}}=o(\frac{1}{\sqrt{n}})
\end{align*}
 as $n\to \infty$ if we choose $\varepsilon=\frac{1}{n^k}$ with $\frac 1 2 < k <\frac 3 4$, where the upper bound does not depend on $m$.  Analogously, it is
\begin{align*}
  \sup_{s\in\R}|  \Phi_{0,1}(r(n)(s+\varepsilon))- \Phi_{0,1}(s)    | \le o(\frac{1}{\sqrt{n}})
\end{align*}
and so
\begin{align*}
\max (  \sup_{s\in\R}|  \Phi_{0,1}(r(n)(s-\varepsilon))- \Phi_{0,1}(s)    | ,   \sup_{s\in\R}|   \Phi_{0,1}(r(n)(s+\varepsilon))- \Phi_{0,1}(s)    |   )\le  o(\frac{1}{\sqrt{n}}).
\end{align*}
For the third summand, set  $\sigma_R^2(n)=\Var(R(n))$ and Chebyshev's inequality yields
\begin{align*}
P\bigg(\bigg|\frac{R(n)}{\sigma(n,t)}\bigg|> \varepsilon\bigg)\le \frac {\sigma^2_R(n)}{ \varepsilon^2\sigma^2(n,t)}= \frac {n\Var(\widetilde X_{1,1})}{  \varepsilon^2\sigma^2(n,t)}=  o(\frac{1}{\sqrt{n}})
\end{align*}
 as $n\to \infty$ and  if we choose again $\varepsilon=\frac{1}{n^k}$ with $\frac 1 2 < k <\frac 3 4$, where the upper bound does not depend on $m$. We deduce that
\begin{align*}
\sup_{s\in\R}| P(S(n,t)\le s)- \Phi_{m\mu_L(n),m\sigma_L^2(n)}(s) |&=\sup_{s\in\R}\bigg| P\bigg(\frac{S(n,t)-m\mu_L(n)}{\sqrt{m}\sigma_L(n)}\le s \bigg)- \Phi_{0,1}(s) \bigg|\\
&\le\frac{1}{\sqrt n}\eta^2+   o(\frac{1}{\sqrt{n}})
\end{align*}
as $n\to \infty$. In all, we have
\begin{align*}
\sup_{s\in\R}| P(S(n,t)\le s|\zeta_{k},k\in\N)- \Phi_{N(t)\mu_L(n),N(t)\sigma_L^2(n)}(s) |
\le \frac{1}{\sqrt n}\eta^2I(N(t)>0)+   o(\frac{1}{\sqrt{n}})
\end{align*}
almost surely  as $n\to \infty$. {Applying the triangle inequality for expectations, we obtain}
\begin{align*}
&\sup_{s\in\R}| P(S(n,t)\le s)-\E( \Phi_{N(t)\mu_L(n),N(t)\sigma_L^2(n)}(s)) |\\
=&\sup_{s\in\R}| \E(P(S(n,t)\le s|\zeta_{k},k\in\N)- \Phi_{N(t)\mu_L(n),N(t)\sigma_L^2(n)}(s)) |\\
\le&\E(\sup_{s\in\R}| P(S(n,t)\le s|\zeta_{k},k\in\N)-  \Phi_{N(t)\mu_L(n),N(t)\sigma_L^2(n)}(s) |)\\
\le&\frac{1}{\sqrt n}\eta^2P(N(t)>0)+   o(\frac{1}{\sqrt{n}})
\end{align*}
 as $n\to \infty$, {which gives the stated unconditional bound}. Note that, due to the right-continuity of any distribution function, we can replace the supremum over $\R$ with the supremum over $\Q$ {to ensure that the supremum is measurable}. 

b)  We show the stated bound conditioned on  $\zeta_{k},$ $k\in\N$, almost surely at first and deduce the {stated unconditional bound} subsequently. For the bound conditioned on  $\zeta_{k},$ $k\in\N$, almost surely we consider  constants  $z_k \in[0,1]$, $k\in\N$, with $\zeta_{k}=z_k,$ $k\in\N$, such that the points in time of the loss events in the counting process $N(t)$ are given by 
$$T_1=\ell_1(z_1),~T_k=\ell_k(z_k,T_1,\dots,T_{k-1}),~k\ge 2.$$ Note  that {the} random variable $N(t)$ is  now constant, $N(t)=m(t)\in\N_0$, say, and  that the arrays of random variables $G_k$, $k\in \N$, {are independent of one another, and each is dissociated}. In the case $N(t)=m(t)=0$ we obtain  $S(n,t)=0\sim N(0,0)=N(m(t)\mu_L(n),m(t)\sigma_L^2(n))$ for all $n\in\N${, so the conditional bound holds with both sides equal to zero.} For that reason, let $N(t)=m(t)\ge 1$.  Remember that $S(n,t)=\sum_{k=1}^{N(t)}L_{k}(n)=\sum_{k=1}^{m(t)}L_{k}(n)$, where $L_{k}(n)=\sum_{i,j=1}^n G_{i,j,k}$, $k\in\N$, are  independent and identically distributed random variables in this situation. It is 
\begin{align*}
\frac{S(n,t)-m(t)\mu_L(n)}{\sqrt{m(t)}\sigma_L(n)}&=\frac{\sum_{k=1}^{m(t)} (L_k(n)-\E(L_k(n)))}{\sqrt{m(t)\Var(L_1(n)-\E(L_1(n)) )}}
\end{align*}
with   independent and identically distributed and centered random variables
\begin{align*}
L_k(n)-\E(L_k(n)),~k=1,\dots,m(t),~n\in\N.
\end{align*}
Let $s\in\R$ be arbitrary. Applying the usual Berry-Esseen theorem,  we have  
\begin{align*}
\Bigg|P\Bigg(\frac{S(n,t)-m(t)\mu_L(n)}{\sqrt{m(t)}\sigma_L(n)}\le s\Bigg)-\Phi_{0,1}(s)\Bigg|\le \frac{1}{\sqrt{m(t)}}\eta^2(n),
\end{align*}
where $\eta^2(n)=\frac{C_{\operatorname{BE}}\E(|L_1(n)-\E(L_1(n))|^3)}{\Var(L_1(n)-\E(L_1(n)))^{3/2}}$, and so
\begin{align*}
\sup_{s\in\R}| P(S(n,t)\le s)- \Phi_{m(t)\mu_L(n),m(t)\sigma_L^2(n)}(s) |&=\sup_{s\in\R}\bigg|P\Bigg(\frac{S(n,t)-m(t)\mu_L(n)}{\sqrt{m(t)}\sigma_L(n)}\le s\Bigg)-\Phi_{0,1}(s)\bigg|\\
&\le  \frac{1}{\sqrt{m(t)}}\eta^2(n).
\end{align*}
We deduce that 
\begin{align*}
\sup_{s\in\R}|P(S(n,t)\le s|\zeta_{k},k\in\N )- \Phi_{N(t)\mu_L(n),N(t)\sigma_L^2(n)}(s)|\le \frac{1}{\sqrt{N(t)}}I(N(t)>0)\eta^2(n)
\end{align*}
almost surely, that is the stated bound conditioned on  $\zeta_{k},$ $k\in\N$, almost surely.  The {stated unconditional bound} follows analogously to the proof of a). 
\end{proof}

\begin{proof} [Proof of Theorem \ref{thm5}.]

a) We recall the expression for the loss of the insurance company:
\begin{align*}
S(n,t)=\sum_{k=1}^{N(t)}L_{k}(n), \quad
L_{k}(n)=\sum_{i,j=1}^nG_{i,j,k}.
\end{align*}
The expectation and variance of $S(n,t)$ are
\begin{align*}
\mu_S(n,t) = & \; \E\big(S(n,t)\big)=\E(N(t)) \mu_L(n)=\lambda t \mu_L(n)\\
\sigma^2_S(n,t)  = & \; \Var\big(S(n,t)\big)=\Var(N(t))\mu_L(n)^2+\E(N(t))\sigma^2_L(n)=\lambda t \big(\mu_L(n)^2+\sigma^2_L(n)\big).
\end{align*}
We can write  
$
S(n,t)-\mu_S(n,t)=\widetilde S(n,t)+R(n,t)
$
with the telescopic sum 
$
\widetilde S(n,t)=\sum_{k=1}^{\lfloor t \rfloor}H_{k}(n)
$
of the increments
\begin{align*}
H_{k}(n)=\sum_{k'=1}^{N(k)}L_{k'}(n)-\sum_{k'=1}^{N(k-1)}L_{k'}(n)-\lambda\mu_L(n), \quad k=1,\dots,\lfloor t \rfloor, 
\end{align*}
and with the remainder term
\begin{align*}
R(n,t)=\sum_{k= N(\lfloor t \rfloor)+1 }^{N(t)}L_{k}(n)-\lambda( t-\lfloor t \rfloor)\mu_L(n).
\end{align*}
Due to Assumptions \ref{a1}--\ref{a4}, and in particular the fact that the counting process $N = (N(t))_{t\in[0,\infty)}$ is a homogeneous Poisson process with intensity $\lambda \in (0,\infty)$, we observe that the increments $H_{1}(n), \dots, H_{\lfloor t \rfloor}(n)$ form a sequence of independent and identically distributed centered random variables. Setting
$
\sigma^2_{H}(n)=\Var\big(H_{1}(n)\big),
$
it follows from the central limit theorem of Lindeberg-L\'evy that 
\begin{align*}
\frac{\widetilde S(n,t)}{\sqrt{\lfloor t \rfloor}\sigma_{H}(n)} \overset{d}{\longrightarrow}N(0,1)~\text{as}~t\to\infty.
\end{align*}
Moreover, it is
\begin{align*}
\sigma^2_{\widetilde S}(n,t)=\Var\big(\widetilde S(n,t)\big)={\lfloor t \rfloor}\Var\big(H_{1}(n)\big)={\lfloor t \rfloor}\sigma^2_{H}(n).
\end{align*}
Due to Assumptions \ref{a1}--\ref{a4}, and in particular the fact that the counting process $N = (N(t))_{t\in[0,\infty)}$ is a homogeneous Poisson process with intensity $\lambda \in (0,\infty)$, the random variables $\widetilde S(n,t)$ and $R(n,t)$ are independent.
For that reason, we have
\begin{align*}
\sigma^2_S(n,t)=\Var\big( S(n,t)-\mu_S(n,t)\big)=\sigma^2_{\widetilde S}(n,t)+\Var\big(R(n,t)\big),
\end{align*}
and, due to the fact that the distribution of $R(n,t)$ is periodic in $t$ with period $1$, it follows that
\begin{align*}
\sigma^2_{\widetilde S}(n,t)\sim\sigma^2_S(n,t)~\text{as}~t\to\infty.
\end{align*}
This implies
\begin{align*}
{\lfloor t \rfloor}\sigma^2_{H}(n)\sim{\sigma^2_S(n,t)}~\text{as}~t\to\infty.
\end{align*}
Furthermore, again due to the fact that the distribution of $R(n,t)$ is periodic in $t$ with period $1$, we have 
$
\frac{R(n,t)}{\sqrt{\lfloor t \rfloor}} \overset{P}{\longrightarrow}0$ as $t\to\infty$.
In summary, Slutsky's theorem implies
\begin{align*}
\frac{S(n,t)-\mu_S(n,t)}{\sigma_S(n,t)}=\frac{\sqrt{\lfloor t \rfloor}\sigma_{H}(n)}{{\sigma_S(n,t)}}\bigg(\frac{\widetilde S(n,t)}{\sqrt{\lfloor t \rfloor}\sigma_{H}(n)}+\frac{1}{\sigma_{H}(n)}\frac{R(n,t)}{\sqrt{\lfloor t \rfloor}}\bigg)
\overset{d}{\longrightarrow}N(0,1)~\text{as}~t\to\infty.
\end{align*}
This completes the proof of a).

b)  Due to our assumptions, we can assume without loss of generality that the random variables   $\widetilde N$, $G_{k}$, $k\in\N$, on the one hand and the  process $T$ on the other hand are defined on different probability spaces $(\Omega_1,\AA_1,P_1)$ and  $(\Omega_2,\AA_2,P_2)$, where  $(\Omega,\AA,P)= (\Omega_1\times\Omega_2,\AA_1\otimes\AA_2,P_1\otimes P_2)$ is endowed with the product measure. We fix an arbitrary $\omega\in\{\lim_{s\to\infty}T(s)=\infty\}\in\AA_2$ and set $t'(t)=T(t)(\omega)$ for $t\in[0,\infty)$; $\lim_{t\to\infty}t'(t)=\infty$ is satisfied. In addition, we write 
\begin{align*}
S'(n,\widetilde N(T(t)))=S(n,t)
\end{align*}
and we fix {an} arbitrary $s\in\R$. Setting
\begin{align*}
h_t(\omega)=P_1\Bigg(\frac{ S'(n,\widetilde N(T(t)(\omega)))-\lambda T(t)(\omega) \mu_L(n)}{\sqrt{\lambda T(t)(\omega) (\mu_L(n)^2+\sigma^2_L(n))}}\le s \Bigg)
\end{align*}
it is
\begin{align*}
h_t(\omega)=P_1\Bigg(\frac{ S'(n,\widetilde N(t'(t)))-\lambda t'(t) \mu_L(n)}{\sqrt{\lambda t'(t) (\mu_L(n)^2+\sigma^2_L(n))}}\le s \Bigg).
\end{align*}
Furthermore,  $S'(n,\widetilde N(t'(t)))$ is the total loss at time $t'(t)$ in the situation in part a) with  expectation
$
\E(S'(n,\widetilde N(t'(t))))=\lambda t'(t) \mu_L(n)
$
and  variance
$
\Var(S'(n,\widetilde N(t'(t))))=\lambda t'(t) \big(\mu_L(n)^2+\sigma^2_L(n)\big).
$
Setting $
h(\omega)=\Phi_{0,1}(s)$ (which is constant in $\omega$), part a) implies that $ h_t(\omega)\rightarrow h(\omega)~\text{as}~t\to\infty$. Because $h_t\ge 0$ and $h_t$ is bounded by 1, it follows from the dominated convergence theorem that 
\begin{align*}
P\Bigg(\frac{ S(n,t)-\lambda T(t) \mu_L(n)}{\sqrt{\lambda T(t) (\mu_L(n)^2+\sigma^2_L(n))}}\le s \Bigg)&=\int h_t(\omega)dP_2(\omega)\\
&=\int h_t(\omega)I\Big(\omega\in\Big\{\lim_{s\to\infty}{T(s)}=\infty\Big\}\Big)dP_2(\omega)\\
&\longrightarrow \int h(\omega)I\Big(\omega\in\Big\{\lim_{s\to\infty}{T(s)}=\infty\Big\}\Big)dP_2(\omega)\\
&=\Phi_{0,1}(s)~\text{as}~t\to\infty,
\end{align*}
observing that $P_2(\{\lim_{s\to\infty}{T(s)}=\infty\})=1$ is satisfied. This completes the proof of part b).
\end{proof}

\color{black}

\begin{proof}[Proof of Theorem \ref{thmrel}.]

For $t\in[0,\infty)$, we consider  the Laplace transform of $(S(n,t),M(n,t))$, {i.e.,}  $\psi_{(S(n,t),M(n,t))}(u,v)=E(e^{-uS(n,t)-vM(n,t)})$ for $(u,v)\in[0,\infty)\times[0,\infty)$, which characterizes the distribution uniquely. With $\eta\in\N_0$, $\iota=(\iota_{i,j,k})_{(i,j,k) \in\N\times \N\times \N}\in\{0,1\}^{\N\times \N\times \N}$, and setting $m=\sum_{i=1}^n\sum_{j=1}^n\sum_{k=1}^{\eta}\iota_{i,j,k}$, we have for $t\in[0,\infty)$ and  $(u,v)\in[0,\infty)\times[0,\infty)$:
\begin{align*}
\E(e^{-uS(n,t)-vM(n,t)}|N(t)=\eta,I=\iota)=\psi_Z(u)^me^{-vm} ,
\end{align*}
thus
\begin{align*}
\psi_{(S(n,t),M(n,t))}(u,v)=\E(\E(e^{-uS(n,t)-vM(n,t)}|N(t),I))=\E(\psi_Z(u)^{M(n,t)}e^{-vM(n,t)})
\end{align*}
on the one hand; on the other hand, we obtain for $m\in\N_0$:
\begin{align*}
\E(e^{-u\sum_{i=1}^{M(n,t)}\widetilde Z_i-vM(n,t)}|M(n,t)=m)=\psi_Z(u)^me^{-vm} ,
\end{align*}
thus
\begin{align*}
\psi_{\big(\sum_{i=1}^{M(n,t)}\widetilde Z_i,M(n,t)\big)}(u,v)=\E(\E(e^{-u\sum_{i=1}^{M(n,t)}\widetilde Z_i-vM(n,t)}|M(n,t)))=\E(\psi_Z(u)^{M(n,t)}e^{-vM(n,t)}).
\end{align*}
This proves the assertion. 
\end{proof}

\end{appendix}

\color{black}

\noindent
\begin{normalsize}{\bf  Acknowledgement}\end{normalsize}

\noindent
The authors thank the anonymous reviewers for their valuable comments and constructive suggestions, which have significantly improved the quality of this work. The authors are grateful to Julian Gerstenberg for valuable discussions.

\noindent
AI-based language tools were used for stylistic editing and clarity, as well as to assist with mathematical proofreading by checking derivations and equations for consistency. All mathematical proofs, calculations, and final results were independently verified by the authors, who take full responsibility for the scientific content and accuracy of the paper.
\color{black}

\noindent
\begin{normalsize}{\bf Competing interests}\end{normalsize}

\noindent
The authors declare that they have no competing interests.

\noindent
\begin{normalsize}{\bf Funding}\end{normalsize}

\noindent
This research received no specific grant from any funding agency in the public, commercial, or not-for-profit sectors.

\noindent
\begin{normalsize}{\bf Data availability}\end{normalsize}

\noindent
This study is based exclusively on simulated data generated for numerical experiments. No external or real-world datasets were used.

\end{document}